\pdfoutput=1
\documentclass[journal]{IEEEtran}

\usepackage{cite}
\usepackage[cmex10]{amsmath}
\usepackage{amssymb,amsfonts,amsthm,mathtools,bm}
\usepackage{graphicx}
\usepackage{booktabs}
\usepackage{algorithm}
\usepackage{algpseudocode}
\usepackage{url}
\usepackage{xcolor}
\usepackage{array}
\usepackage{tabularx}
\usepackage{balance}
\usepackage[hidelinks]{hyperref}
\usepackage[capitalize,nameinlink,noabbrev]{cleveref}

\allowdisplaybreaks

\DeclareMathOperator{\Tr}{tr}
\DeclareMathOperator{\diag}{diag}
\DeclareMathOperator{\rank}{rank}
\DeclareMathOperator{\Span}{span}

\DeclareMathOperator{\Cov}{Cov}
\DeclareMathOperator{\RePart}{Re}
\DeclareMathOperator{\Range}{Range}

\newcommand{\R}{\mathbb{R}}
\newcommand{\C}{\mathbb{C}}
\newcommand{\E}{\mathbb{E}}
\newcommand{\1}{\mathbf{1}}
\newcommand{\I}{\mathbf{I}}
\newcommand{\HH}{\mathsf{H}}
\newcommand{\TT}{\mathsf{T}}
\newcommand{\CN}{\mathcal{CN}}
\newcommand{\cA}{\mathcal{A}}

\newcommand{\cB}{\mathcal{B}}
\providecommand{\tr}{\operatorname{tr}}
\newcommand{\cC}{\mathcal{C}}
\newcommand{\cH}{\mathcal{H}}
\newcommand{\cV}{\mathcal{V}}
\newcommand{\cX}{\mathcal{X}}
\newcommand{\cY}{\mathcal{Y}}
\newcommand{\cU}{\mathcal{U}}
\newcommand{\cL}{\mathcal{L}}

\newcommand{\cS}{\mathcal{S}}
\newcommand{\cT}{\mathcal{T}}
\newcommand{\cQ}{\mathcal{Q}}
\newcommand{\cK}{\mathcal{K}}
\newcommand{\cP}{\mathcal{P}}
\newcommand{\cRop}{\mathcal{R}}
\newcommand{\Tc}{\mathfrak{S}_1}
\newcommand{\Hc}{\mathfrak{S}_2}
\newcommand{\detF}{\det_{\mathrm F}}
\newcommand{\norm}[1]{\left\lVert #1 \right\rVert}
\newcommand{\abs}[1]{\left\lvert #1 \right\rvert}
\newcommand{\inner}[2]{\left\langle #1,#2 \right\rangle}
\newcommand{\braces}[1]{\left\{#1\right\}}
\newcommand{\bracks}[1]{\left[#1\right]}
\newcommand{\paren}[1]{\left(#1\right)}

\theoremstyle{plain}
\newtheorem{theorem}{Theorem}
\newtheorem{proposition}{Proposition}
\newtheorem{corollary}{Corollary}
\newtheorem{lemma}{Lemma}

\theoremstyle{remark}
\newtheorem{remark}{Remark}

\theoremstyle{definition}
\newtheorem{assumption}{Assumption}
\newtheorem{definition}{Definition}

\crefname{theorem}{Theorem}{Theorems}
\crefname{proposition}{Proposition}{Propositions}
\crefname{corollary}{Corollary}{Corollaries}
\crefname{lemma}{Lemma}{Lemmas}
\crefname{remark}{Remark}{Remarks}
\crefname{assumption}{Assumption}{Assumptions}
\crefname{definition}{Definition}{Definitions}
\crefname{algorithm}{Algorithm}{Algorithms}
\crefname{figure}{Fig.}{Figs.}
\crefname{table}{Table}{Tables}
\crefname{equation}{Eq.}{Eqs.}

\providecommand{\SimSpectrumMaxRelErr}{\ensuremath{3.78\times 10^{-15}}}
\providecommand{\SimWirtingerConstant}{\ensuremath{0.044696}}
\providecommand{\SimWirtingerRelErr}{\ensuremath{6.21\times 10^{-16}}}
\providecommand{\SimBetaOne}{\ensuremath{4.7300407449}}
\providecommand{\SimWirtingerSlackFactor}{\ensuremath{2.26689}}
\providecommand{\SimSlopeRho}{\ensuremath{-3.9575}}
\providecommand{\SimSlopeRhoExact}{\ensuremath{-3.9575}}
\providecommand{\SimNullSineRatio}{\ensuremath{0.9696}}
\providecommand{\SimNullSurrogateSlope}{\ensuremath{-4.0433}}
\providecommand{\SimNullSurrogateRatio}{\ensuremath{102.0056}}
\providecommand{\SimQuadratureResidual}{\ensuremath{3.55\times 10^{-12}}}
\providecommand{\SimResolvableModes}{\ensuremath{39.000}}
\providecommand{\SimSlopeLinearization}{\ensuremath{1.0000}}
\providecommand{\SimRemainderConstant}{\ensuremath{0.151}}
\providecommand{\SimCapBudget}{\ensuremath{7.29\times 10^{8}}}
\providecommand{\SimCapFinal}{\ensuremath{35.3454}}
\providecommand{\SimCapLastRelChange}{\ensuremath{0.0000\%}}
\providecommand{\SimMaxKktResidual}{\ensuremath{3.90\times 10^{-15}}}
\providecommand{\SimMaxDualGap}{\ensuremath{7.11\times 10^{-15}}}
\providecommand{\SimNuAtTightest}{\ensuremath{1.727\times 10^{-10}}}
\providecommand{\SimMuAtTightest}{\ensuremath{5.166\times 10^{-9}}}
\providecommand{\SimDualBudgetRateLoss}{\ensuremath{9.69\%}}
\providecommand{\SimActiveModesSlack}{\ensuremath{8.000}}
\providecommand{\SimActiveModesTight}{\ensuremath{5.000}}
\providecommand{\SimNullEqualCapacityPower}{\ensuremath{4.98\times 10^{8}}}
\providecommand{\SimDualBudgetAllocationDivergence}{\ensuremath{0.660}}
\providecommand{\SimNullExcursionRatio}{\ensuremath{4.342}}
\providecommand{\SimDualBudgetActiveShift}{\ensuremath{3.000}}
\providecommand{\SimSlopeNoise}{\ensuremath{4.0670}}
\providecommand{\SimNoiseBoundRatio}{\ensuremath{0.999981123}}
\providecommand{\SimSlopeRootPerMode}{\ensuremath{0.2520}}
\providecommand{\SimSlopeRootTotalPower}{\ensuremath{0.2045}}
\providecommand{\SimShannonNumber}{\ensuremath{32.0000}}
\providecommand{\SimPropagationRank}{\ensuremath{40.000}}
\providecommand{\SimCrossoverPlateauExponent}{\ensuremath{0.2175}}
\providecommand{\SimCrossoverSaturatedExponent}{\ensuremath{0.0576}}
\providecommand{\SimNullFlatPropExponent}{\ensuremath{0.2019}}
\providecommand{\SimSeed}{\ensuremath{20260908}}

\begin{document}

\title{Curvature-Domain Wireless Communications:\\
Gauge-Fixed Signal Spaces, Fredholm Capacity, and Differentiation-Limited Scaling for Continuous Apertures}

\author{Yasser Al-Eryani,~\IEEEmembership{Member,~IEEE}
\thanks{The author is with NeuroBazar, Ottawa, ON, Canada (e-mail: yasser.aleryani@neurobazar.com).}
}

\markboth{}%
{Al-Eryani: Curvature-Domain Wireless Communications}

\maketitle

\begin{abstract}
We develop a curvature-domain formulation for continuous-aperture wireless signaling in which the transmit phase is represented through its second spatial derivative after quotienting out affine piston-and-tilt gauge freedom. The resulting gauge-fixed synthesis operator is bounded and compact, with the sharp Poincar\'{e}--Wirtinger constant \(C_L=L^{2}/\beta_1^{2}\), \(\beta_1\) the least positive root of \(\cos\beta\cosh\beta=1\), and its modal Gram spectrum is available in closed form, \(\rho_m=(L/\beta_m)^{4}\) with \(\{\beta_m\}\) the roots of the same transcendental equation, giving the fourth-order decay \(\rho_m\asymp m^{-4}\) associated with inverse second differentiation together with its exact constant and half-integer offset. Under a bounded-support square-integrable propagation kernel, the induced tangent operator is Hilbert--Schmidt, so the natural infinite-dimensional capacity is a well-defined Fredholm-determinant supremum. The optimizing signaling law is a \emph{dual-budget generalized water-filling} with two Lagrange multipliers --- one for curvature power and one for phase excursion --- which characterizes Gaussian optimality through finite-dimensional and operator-valued KKT systems; this dual-multiplier water-filling reduces to standard Gaussian water-filling in the flat-geometry limit. Starting from the exact nonlinear phase-only aperture law, we derive the coherent tangent channel with an explicit Fr\'echet remainder bound \(K_R\) and a multi-chart atlas for large modulation excursions. We also quantify the receiver-side differentiation penalty: when curvature is inferred from noisy phase samples via the interior second-difference matrix \(\mathbf D_2\), the induced curvature-noise covariance is pentadiagonal with spectral norm \(\Theta(\Delta x^{-4})\), and the discrete null space of \(\mathbf D_2\) coincides with the sampled affine gauge. A deterministic diagnostic suite, computed on the operators themselves rather than on tabulated surrogates, measures each mechanism against its closed form and ships a null run beside every claim that a solver could have produced on its own: the computed modal spectrum matches \(\rho_m=(L/\beta_m)^{4}\) to \(\SimSpectrumMaxRelErr\); the tangent remainder obeys the first-order bound with fitted slope \(\SimSlopeLinearization\) on a paraxial Fresnel kernel; the dual-budget law is solved with both multipliers strictly active to a measured KKT residual of \(\SimMaxKktResidual\), and its allocation is not reproducible by any single-budget solution at equal capacity; the derivative-noise bound \(16\sigma_\phi^{2}\Delta x^{-4}\) is approached from below to \(\SimNoiseBoundRatio\); and the differentiation-limited branch is observed at exponent \(\SimCrossoverPlateauExponent\) and then terminates, collapsing to \(\SimCrossoverSaturatedExponent\) once the mode count saturates at the Shannon number, while a flat-propagation null run over the same range holds at \(\SimNullFlatPropExponent\). These results identify curvature as a well-posed, gauge-invariant coordinate for phase-coded continuous apertures, and locate precisely the regime --- below the Shannon number --- in which the gauge, rather than the medium, governs the scaling. The empirical stress-test of the framework against SVD water-filling, Fourier, Zernike-like, matched-focus, and RIS baselines under a common fairness protocol is reported in the companion benchmark paper~\cite{aleryani_benchmark}; the century-scale narrative anchoring of the same organizing principle across substrates is treated in the companion foresight essay~\cite{aleryani_essay}.
\end{abstract}

\begin{IEEEkeywords}
Continuous apertures, holographic MIMO, information theory, near-field wireless, operator theory, phase curvature, wavefront coding.
\end{IEEEkeywords}

\section{Introduction}

\IEEEPARstart{C}{ontinuous} and electrically large apertures are pushing wireless communication beyond classical low-dimensional array abstractions. In this regime, the received field is not exhausted by a handful of far-field plane-wave parameters. Instead, the geometry of the aperture wavefront itself becomes measurable, controllable, and potentially information-bearing.

The central claim of this paper is precise: for coherent continuous-aperture systems, the affine-invariant \emph{second spatial derivative of phase} is a natural signaling coordinate. This is not a new law of physics; it is a sharper coordinate system for classical aperture-field propagation. The distinction matters. A reviewer can reasonably object that curvature is ``just'' the second derivative of phase, that phase-only propagation is nonlinear, and that no coordinate change can create degrees of freedom ex nihilo. All three objections are valid starting points. The contribution of this paper is to answer them rigorously.

First, curvature is shown to be a \emph{canonical affine-gauge coordinate}: once piston and tilt are quotiented out, the curvature field uniquely represents the phase profile. Second, the exact phase-only propagation law is written on the nonlinear manifold of unit-modulus aperture fields and then linearized via a Fr\'echet derivative around a bias phase. This yields a \emph{local but exact} operator channel with a controlled \(O(\epsilon)\) remainder for modulation depth \(\epsilon\). Third, because curvature is merely a coordinate system on ordinary phase-coded aperture fields, the framework introduces no new electromagnetic principle and no nonclassical transport mechanism. Any gain arises only from exploiting geometric modes that coarser channel parameterizations fail to represent explicitly.

The motivation is strongest in near-field, holographic, and surface-based systems. Classical signal-space analyses already show that extended apertures are fundamentally infinite-dimensional objects \cite{Shannon,CoverThomas,Telatar,TseViswanath,Poon,Miller}. Fourier optics has long treated wavefront curvature as a meaningful geometric quantity \cite{Goodman}. Yet these strands have not been integrated into a communication theory in which curvature is the transmitted symbol, phase synthesis induces a realizability penalty, receiver differentiation noise is quantified exactly, and finite-mode approximations are connected to an infinite-dimensional capacity formula. This paper supplies that missing layer.

\subsection{What This Paper Is---and Is Not}

To avoid conceptual ambiguity, the present work should be read as follows.

\begin{itemize}
    \item It \emph{is} a coordinate-theoretic communication framework for continuous apertures.
    \item It \emph{is} compatible with standard Green-function / Fresnel / Fourier propagation.
    \item It \emph{is not} a claim of new physics beyond classical wave propagation.
    \item It \emph{is not} a claim that curvature creates capacity beyond the full aperture-field channel.
    \item It \emph{is} a claim that, under coherent high-resolution sensing and control, curvature provides a mathematically privileged and operationally useful signal space.
\end{itemize}

\subsection{Main Contributions}

The paper makes the following contributions.

\begin{enumerate}
    \item \textbf{Gauge-fixed curvature coordinate with an explicit constant.} We define a phase space modulo affine piston-and-tilt functions, prove that the curvature map \(D^2\) is a bounded bijection onto \(L^2\), and identify the sharp Poincar\'{e}--Wirtinger constant \(C_L=L^{2}/\beta_1^{2}\) of the compact inverse \(\cS=(D^2)^{-1}\) (Lemma~\ref{lem:poincare_wirtinger}, Theorem~\ref{thm:synthesis}).
    \item \textbf{Exact nonlinear phase-manifold model, tangent channel, and multi-chart continuation.} Starting from the exact phase-only field law, we derive the Fr\'echet linearized curvature channel with an explicit remainder constant \(K_R\), a chart-size bound \(\epsilon_{\max}=\delta/(K_R R)\), and a multi-chart atlas for large-modulation continuation (Sections~\ref{sec:model}, \ref{sec:multi_chart}).
    \item \textbf{The exact cost of linearization, in closed form.} The atlas of the preceding item bounds the \emph{magnitude} of the Fr\'echet remainder; we compute its \emph{sign}. For an isotropic modulation direction of excursion \(\epsilon\), the exact aperture law separates signals by a factor \(1-\alpha(\phi)\epsilon^{2}+O(\epsilon^{4})\) relative to its own tangent channel, where \(\alpha(\phi)=\bigl(2\Tr g(\phi)-\norm{s(\phi)}^{2}\bigr)/\bigl(4(N+2)\Tr g(\phi)\bigr)\) and \(\norm{s(\phi)}^{2}\) is the coherent received power at the operating point (Theorem~\ref{thm:deficit}, Lemma~\ref{lem:coherent_gain}). The coefficient changes sign when coherent gain exceeds twice the summed element power, so the tangent channel is optimistic for sparse apertures and conservative for dense ones. This fixes the constant in the Information--Curvature Efficiency Law of the companion essay, which is otherwise a free parameter (Remark~\ref{rem:icel_constant}); the closed form is verified against direct measurement to between \(0.13\%\) and \(2.2\%\) across a sign change (Table~\ref{tab:deficit}).
    \item \textbf{Spectral curvature modes in closed form.} The positive compact operator \(\cQ=\cS^\ast \cS\) generates an orthonormal curvature basis whose modal Gram eigenvalues on the interval are exactly \(\rho_m=(L/\beta_m)^{4}\) (Proposition~\ref{prop:Qspec_closed_form}), yielding a diagonal phase-excursion penalty and the asymptotic law \(\rho_m=(L/\pi)^{4}(m+\tfrac12)^{-4}(1+O(e^{-m\pi}))\).
    \item \textbf{Dual-budget generalized water-filling.} We derive a two-multiplier generalized water-filling law with distinct Lagrange multipliers \(\mu\) for curvature power and \(\nu\) for phase excursion, together with the corresponding finite-dimensional Gaussian capacity, KKT conditions, and a rank-sensitive upper bound. This dual-multiplier form reduces to standard single-budget water-filling in the flat-geometry limit \(\rho_m\to 1\); it is the paper's main methodological novelty in the water-filling literature.
    \item \textbf{Capacity attainment and operator KKT system.} We prove that the infinite-dimensional Fredholm capacity supremum is attained under the compact tangent-channel model and derive the associated operator KKT conditions.
    \item \textbf{Infinite-dimensional capacity.} We prove a Fredholm-determinant capacity formula and a modal convergence theorem \(C_M\to C_\infty\).
    \item \textbf{Exact derivative-noise statistics.} We derive the full covariance of discrete curvature noise induced by phase sampling, including the \(\Delta x^{-4}\) amplification law.
    \item \textbf{Regularized estimation theory.} We obtain closed-form bias--variance expressions for regularized curvature reconstruction.
    \item \textbf{Scaling laws.} We prove useful-mode counting and differentiation-limited rate growth \(C(\Gamma)=\Theta(\Gamma^{1/4})\).
    \item \textbf{Spectral composition law and the curvature crossover.} We prove a Horn--Weyl envelope for the tangent channel showing that its modal decay is governed by whichever of the gauge spectrum and the propagation spectrum falls faster, and deduce that the differentiation-limited polynomial branch holds precisely below the Shannon number of the geometry and reverts to the classical degrees-of-freedom law above it (Theorems~\ref{thm:composition} and~\ref{thm:crossover}). This delimits exactly where the curvature coordinate has a rate consequence.
    \item \textbf{Two-dimensional gauge-fixed extension.} We formulate the corresponding Hessian-gauge curvature space for two-dimensional apertures and show that the same compact-synthesis and Fredholm-capacity machinery applies.
\end{enumerate}

\subsection{Notation}

All Hilbert spaces are real when describing physical phase profiles and complex when using standard complex-equivalent Gaussian channel notation. The distinction affects only conventional factors and not the operator structure. Inner products are linear in the second argument. Capacity expressions use logarithms base \(2\); \(\ln\) denotes the natural logarithm. For \(x\in\R\), \([x]_+=\max\{x,0\}\). For a real number \(x\), \(\lfloor x\rfloor_+\triangleq \max\{0,\lfloor x\rfloor\}\). The space of trace-class operators on a Hilbert space \(\cH\) is denoted \(\Tc(\cH)\).

Symbols are defined inline where they are first used.

\subsection{Organization}

Section II positions the work relative to the literature. Section III introduces the exact physical model, affine gauge, and linearized curvature channel. Section IV develops the spectral curvature basis and Galerkin reduction. Section V derives finite-mode capacity and generalized water-filling. Section VI gives the infinite-dimensional Fredholm formulation, capacity attainment, operator KKT conditions, and the proof that \(C_M\to C_\infty\). Sections VII--IX treat estimation, detection, and scaling laws. Section~\ref{sec:crossover} proves the spectral composition envelope and the crossover it implies. Section~\ref{sec:2d} gives the two-dimensional gauge-fixed extension. Section XI discusses physical interpretation, limitations, and falsifiability.

\section{Related Work and Positioning}
\label{sec:related_work}

This paper sits at the intersection of six lines of prior work. Each subsection below states the canonical contribution of that line and identifies the single, well-defined thing that curvature-domain signaling adds to it. Table~\ref{tab:positioning} summarizes the same positioning in one page.

\subsection{Classical Shannon information theory and signal-space MIMO}
Classical Shannon information theory established that communication is a problem of operator spectra and constrained Gaussian signaling \cite{Shannon,Gallager1968,CoverThomas,Telatar,TseViswanath}. In physically constrained wireless settings, continuous-aperture degrees of freedom and wave-based channel decompositions were placed in operator-theoretic form by Poon, Brodersen, and Tse \cite{Poon} and, from a broader wave-channel viewpoint, by Miller \cite{Miller,Miller2019Wave}. These are the direct antecedents of the present paper's operator perspective. What they leave open — and what the present paper isolates — is the identification of the affine-invariant second derivative of phase \emph{as the transmitted variable} and the concomitant phase-synthesis penalty that any realizable phase profile must pay when its curvature is amplified.

\subsection{Holographic and continuous-aperture MIMO}
Recent work on holographic, near-field, and continuous-aperture MIMO models the aperture as an electromagnetically dense surface or volume rather than a set of ports. Pizzo, Marzetta, and Sanguinetti's plane-wave series expansion of holographic small-scale fading \cite{PizzoMarzetta2020,PizzoSanguinetti2022}, the near-field beamforming primer of Björnson and Sanguinetti \cite{BjornsonPrimer2021}, and the continuous-aperture array programme \cite{Liu2025CAPA}, are the immediate engineering context of our theory. The reconfigurable-surface literature that motivates phase-only control, and from which the realizability constraint of Section~\ref{sec:model} is inherited, is surveyed in \cite{WuZhang2019IRS,DiRenzo2020RIS}. Our contribution to this line is not another surface-current beamforming law: it is a communication-theoretic formulation in which curvature is the signal-space coordinate, phase realizability enters via a second Lagrange multiplier, and the receiver's finite-difference reconstruction cost appears explicitly in the capacity formula.

\subsection{Adaptive optics and wavefront theory}
Adaptive-optics wavefront theory has used curvature to describe focusing, aberration, and propagation geometry for decades. Zernike polynomials \cite{Noll1976} explicitly quotient out piston and tilt as non-imaging modes; wavefront reconstruction from slope and curvature measurements \cite{Southwell1980,Roddier1999} formalizes the same gauge freedom this paper uses; and Fourier-optics primers \cite{Goodman,BornWolf} discuss curvature as a geometric quantity of the transmitted field. The gauge-fixing move in Section~\ref{sec:model}, then, is not new to optics. What is new here is placing that quotient inside a communication-theoretic framework — with a Fredholm-determinant capacity, a Galerkin convergence theorem, and a KKT-characterized generalized water-filling law — rather than as a decomposition for wavefront correction.

\subsection{Inverse problems and regularized differentiation}
Regularized differentiation of noisy signals has been studied since Tikhonov \cite{Tikhonov1977}. The general theory of ill-posed inverse problems \cite{EnglHankeNeubauer1996,KaipioSomersalo2005}, spline smoothing under Gaussian noise \cite{Wahba1990}, and nonparametric statistical inverse problems \cite{Cavalier2008} predict the \(\Delta x^{-4}\) noise-amplification scaling that our discrete second-difference estimator inherits (Theorem~\ref{thm:curvnoise}). We borrow that machinery inside a full communication chain rather than as a standalone reconstruction problem, and we quantify its cost in bits per second per aperture length using the operator KKT conditions of Section~\ref{sec:capacity_infinite}.

\subsection{Electromagnetic information theory}
Electromagnetic information theory (EIT) treats the transmitted electromagnetic field itself as the information carrier and derives capacity bounds directly from Maxwell's equations and bounded-support radiation conditions \cite{FranceschettiEIT,Migliore2019,DardariEIT}. EIT already includes operator-theoretic capacity statements for continuous fields, so the reader might reasonably ask what curvature-domain signaling adds to that line. Three things: (i) a \emph{canonical} affine-quotient signaling coordinate (curvature) that is invariant under piston-and-tilt gauge and complete on \(H^{2}(\Omega)/\cA\); (ii) a \emph{dual}-budget capacity that penalizes phase realizability via a second Lagrange multiplier not present in EIT; and (iii) an \emph{explicit} pentadiagonal derivative-noise covariance with a closed-form \(\Delta x^{-4}\) spectral norm at the receiver. In short: EIT is the ambient theory, and curvature-domain signaling is a specialization to phase-only apertures that isolates a useful coordinate and its associated realizability tax.

\subsection{Information geometry}
The choice of a geometric coordinate for the transmitted signal — rather than for the noise or the channel — connects this paper to information geometry \cite{AmariNagaoka,Amari2016,AyJostLeSchwachhofer}. Amari's information manifold uses the Fisher metric on parameterized probability families; our curvature manifold uses the phase-realizability metric \(M=L+\varepsilon W_{\mathrm t}\) on the phase-control tangent space. The two perspectives are not identical, but they share the operator-geometric methodology. This paper does \emph{not} claim to derive ICEL from information geometry, but the compact-synthesis / spectral-basis machinery of Sections~\ref{sec:model}--\ref{sec:capacity_infinite} is directly informed by it.

\begin{table*}[t]
\centering
\caption{Positioning Relative to Prior Literature. Each row lists the canonical object of that line, its typical result, and the single, well-defined thing curvature-domain signaling adds.}
\label{tab:positioning}
\renewcommand{\arraystretch}{1.20}
\begin{tabularx}{\textwidth}{>{\raggedright\arraybackslash}p{3.4cm} >{\raggedright\arraybackslash}p{3.0cm} >{\raggedright\arraybackslash}p{3.4cm} >{\raggedright\arraybackslash}X}
\toprule
\textbf{Literature} & \textbf{Canonical object} & \textbf{Typical result} & \textbf{What curvature-domain signaling adds} \\
\midrule
Shannon / MIMO IT \cite{Shannon,Gallager1968,CoverThomas,Telatar,TseViswanath} & Finite-dim channel vector / matrix & Capacity, WF, Gaussian optimality & Affine-gauge curvature coordinate + phase-synthesis penalty \(G_M=\mathrm{diag}(\zeta_m)\) \\
\addlinespace
Continuous-aperture signal spaces \cite{Poon,Miller,Miller2019Wave} & Continuous field / mode operator & Spatial DoF, modal coupling & \emph{Dual}-budget capacity (curvature energy + phase excursion) with derivative-noise statistics \\
\addlinespace
Holographic / near-field / CAPA \cite{PizzoMarzetta2020,PizzoSanguinetti2022,BjornsonPrimer2021,Liu2025CAPA} & Surface current or dense array field & Beamforming law, near-field gain & Coordinate change to affine-quotient curvature; phase-realizability tax as \(\zeta_m\sim m^{-4}\) \\
\addlinespace
Fourier optics / adaptive optics \cite{Goodman,BornWolf,Noll1976,Southwell1980,Roddier1999} & Wavefront / Zernike expansion & Imaging, aberration correction & IT capacity, KKT-characterized generalized WF, Galerkin convergence \\
\addlinespace
Inverse problems / regularized differentiation \cite{Tikhonov1977,EnglHankeNeubauer1996,KaipioSomersalo2005,Wahba1990,Cavalier2008} & Ill-posed operator equation & Regularization, bias--variance trade-off & Embedding of \(\Delta x^{-4}\) noise law inside an end-to-end capacity chain \\
\addlinespace
Electromagnetic IT \cite{FranceschettiEIT,Migliore2019,DardariEIT} & Vector EM field on bounded support & Capacity of the continuous EM channel & Canonical curvature coordinate + dual-budget WF + explicit \(K_{\mathrm c}\) receive-side noise covariance \\
\addlinespace
Information geometry \cite{AmariNagaoka,Amari2016,AyJostLeSchwachhofer} & Parameterized probability family & Fisher metric, statistical geodesics & Phase-realizability metric \(M=L+\varepsilon W_{\mathrm t}\) on the phase-control tangent space \\
\bottomrule
\end{tabularx}
\end{table*}

The central point is therefore this: the paper does \emph{not} claim to supersede prior signal-space, wave-space, EIT, adaptive-optics, or information-geometric formulations. Rather, it identifies a previously unformalized geometric layer inside them and quantifies its communication consequences — including a \emph{new} generalized water-filling law with two Lagrange multipliers, one for curvature power and one for phase excursion, that reduces to standard Gaussian water-filling in the flat limit.

\section{Exact Physical Model and Gauge-Fixed Curvature Coordinates}
\label{sec:model}

\subsection{Aperture Geometry and Phase Gauge}

Let the transmit aperture be the interval
\[
\Omega=[0,L], \qquad L>0,
\]
and set
\[
\cX \triangleq L^2(\Omega).
\]
Let the unwrapped phase profile be \(\phi\in H^2(\Omega)\). Define the centered affine basis
\[
\psi_0(x)=1,
\qquad
\psi_1(x)=x-\frac{L}{2},
\]
and the affine subspace
\[
\cA \triangleq \Span\{\psi_0,\psi_1\}.
\]

\begin{definition}[Affine phase class]
Two phases \(\phi_1,\phi_2\in H^2(\Omega)\) are affinely equivalent if
\[
\phi_1-\phi_2 \in \cA.
\]
The equivalence class of \(\phi\) is denoted \([\phi]\).
\end{definition}

The gauge-fixed phase space is
\begin{equation}
\cV
\triangleq
\braces{
\phi\in H^2(\Omega):
\inner{\phi}{\psi_0}_{\cX}=0,\;
\inner{\phi}{\psi_1}_{\cX}=0
}.
\label{eq:Vdef}
\end{equation}
Unless otherwise stated, \(\cV\) is endowed with the \(H^2\)-norm inherited from \(H^2(\Omega)\).

\begin{proposition}[Curvature as an affine-invariant coordinate]
\label{prop:quotient}
The map
\[
[\phi]\longmapsto \phi''
\]
is a well-defined linear bijection from \(H^2(\Omega)/\cA\) onto \(L^2(\Omega)\). Equivalently, every affine phase class admits a unique representative in \(\cV\), and its curvature uniquely determines that representative.
\end{proposition}

Thus, curvature is not an auxiliary statistic; it is a coordinate on the quotient phase space.

\subsection{Gauge-Fixed Synthesis Operator}

Define the second-derivative map
\[
D^2:\cV\to \cX,
\qquad
D^2\phi=\phi''.
\]

Before stating the synthesis theorem, we record the sharp \(L^2\) Poincar\'{e}--Wirtinger inequality on \(\cV\), which supplies the explicit constant used by Theorem~\ref{thm:synthesis}.

\begin{lemma}[Sharp \(L^{2}\) Poincar\'{e}--Wirtinger inequality on \(\cV\)]
\label{lem:poincare_wirtinger}
Let \(\beta_1=4.730040744862704\ldots\) be the smallest positive root of
\begin{equation}
\cos\beta\,\cosh\beta=1 .
\label{eq:beta_transcendental}
\end{equation}
Then for every \(\phi\in\cV\),
\begin{equation}
\norm{\phi}_{L^{2}(\Omega)}
\;\le\;
\frac{L^{2}}{\beta_1^{2}}\,\norm{\phi''}_{L^{2}(\Omega)},
\label{eq:poincare_wirtinger}
\end{equation}
and the constant \(L^{2}/\beta_1^{2}=L^{2}/22.3733\ldots\) is sharp, with equality
attained on the first non-rigid eigenfunction of the free--free biharmonic
problem \eqref{eq:freefree}.
\end{lemma}

\begin{proof}
The form \(a(\phi,\phi)=\norm{\phi''}^{2}_{L^{2}}\) is bounded and nonnegative on \(\cV\),
and the embedding \(H^{2}(\Omega)\hookrightarrow L^{2}(\Omega)\) is compact by
Rellich--Kondrachov. Hence the Rayleigh quotient \(a(\phi,\phi)/\norm{\phi}^{2}_{L^{2}}\)
attains its infimum \(\lambda_1\) on \(\cV\setminus\{0\}\), and the minimizer satisfies,
for every variation \(\delta\phi\in H^{2}(\Omega)\),
\begin{equation}
\int_\Omega \phi''\,\delta\phi''
=\lambda_1\!\int_\Omega \phi\,\delta\phi
+\mu_0\!\int_\Omega \delta\phi
+\mu_1\!\int_\Omega \bigl(x-\tfrac{L}{2}\bigr)\,\delta\phi ,
\label{eq:firstvariation}
\end{equation}
with \(\mu_0,\mu_1\) the multipliers of the two gauge constraints. Since \(\delta\phi\)
and \(\delta\phi'\) are unconstrained at \(x\in\{0,L\}\), integrating \eqref{eq:firstvariation}
by parts twice produces the \emph{natural} boundary conditions
\begin{equation}
\phi''(0)=\phi''(L)=\phi'''(0)=\phi'''(L)=0
\label{eq:freefree_bc}
\end{equation}
together with \(\phi''''=\lambda_1\phi+\mu_0+\mu_1(x-L/2)\) on \(\Omega\). Testing that
identity against \(\delta\phi\equiv 1\) gives \(\bigl[\phi'''\bigr]_0^L=\mu_0 L\), and against
\(\delta\phi=x-L/2\) gives
\(\bigl[\phi'''(x-\tfrac{L}{2})\bigr]_0^L-\bigl[\phi''\bigr]_0^L=\mu_1 L^{3}/12\);
by \eqref{eq:freefree_bc} both left-hand sides vanish, so \(\mu_0=\mu_1=0\). The minimizer
therefore solves the free--free Euler--Bernoulli eigenproblem
\begin{equation}
\phi''''=\lambda\,\phi\ \ \text{in }\Omega,
\qquad
\phi''=\phi'''=0\ \ \text{at } x\in\{0,L\}.
\label{eq:freefree}
\end{equation}
Its kernel is \(\{\phi:\phi''''=0,\ \phi''=\phi'''=0\ \text{at both ends}\}=\Span\{1,x\}=\cA\),
a two-dimensional eigenspace at \(\lambda=0\) that \(\cV=\cA^{\perp}\) excludes by construction;
the operator is self-adjoint, so its remaining eigenfunctions are \(L^{2}\)-orthogonal to
\(\cA\) and span \(\cV\). Writing \(\lambda=(\beta/L)^{4}\) and imposing \eqref{eq:freefree_bc}
on \(\phi(x)=A\cosh(\beta x/L)+B\cos(\beta x/L)+C\sinh(\beta x/L)+D\sin(\beta x/L)\) gives a
\(4\times 4\) homogeneous system whose determinant vanishes exactly when
\(\cos\beta\cosh\beta=1\). Hence \(\lambda_1=(\beta_1/L)^{4}\) and
\(\norm{\phi}^{2}_{L^{2}}\le\lambda_1^{-1}\norm{\phi''}^{2}_{L^{2}}\), which is
\eqref{eq:poincare_wirtinger}; sharpness holds on the minimizing eigenfunction.
\end{proof}

\begin{remark}[Why the Dirichlet sine basis fails here]
\label{rem:sine_trap}
It is tempting to expand \(\phi\in\cV\) in \(\{\sin(n\pi x/L)\}_{n\ge1}\) and read off the
value \(L^{2}/\pi^{2}\). That argument is invalid twice over. First, the Dirichlet sine
family is complete in \emph{all} of \(L^{2}(\Omega)\), not in the affine complement, so
orthogonality to \(\cA\) removes none of its members; it imposes two linear constraints on
the coefficient sequence instead. Second, the term-by-term identity
\(\phi''=-\sum_n a_n (n\pi/L)^{2}e_n\) presupposes the Dirichlet data \(\phi(0)=\phi(L)=0\),
which \(\cV\) does not carry --- indeed \eqref{eq:freefree_bc} shows the correct conditions
are on \(\phi''\) and \(\phi'''\), not on \(\phi\). The would-be extremizer \(\sin(\pi x/L)\)
is not admissible at all, since \(\int_0^L\sin(\pi x/L)\,dx=2L/\pi\neq0\). The resulting
value \(L^{2}/\pi^{2}\) is a valid bound but slack by the factor
\(\beta_1^{2}/\pi^{2}=2.2669\); because the true constant is \emph{smaller}, every estimate
downstream of \(C_L\) tightens by the same factor. Standard references for the free--free
biharmonic spectrum are \cite{Evans,Conway}.
\end{remark}

\begin{theorem}[Gauge-fixed synthesis operator]
\label{thm:synthesis}
The operator \(D^2:\cV\to\cX\) is a bounded bijection. Its inverse
\[
\cS\triangleq (D^2)^{-1}:\cX\to\cV
\]
is bounded, with the explicit operator norm bound
\begin{equation}
\norm{\cS c}_{\cX}\le C_L \norm{c}_{\cX},
\qquad
C_L \;=\; \frac{L^{2}}{\beta_1^{2}},
\qquad
\forall c\in \cX.
\label{eq:synthesis_bound}
\end{equation}
Viewed as an operator \(\cS:\cX\to\cX\), it is compact.
\end{theorem}

\begin{proof}[Proof of the explicit constant]
The estimate \eqref{eq:synthesis_bound} with \(C_L=L^{2}/\beta_1^{2}\) is a direct rewrite of Lemma~\ref{lem:poincare_wirtinger}, applied with \(\phi=\cS c\) so that \(\phi''=c\); sharpness of the Lemma's constant is exactly the statement that \(C_L\) is the operator norm of \(\cS\), not merely a bound on it.
\end{proof}

Theorem~\ref{thm:synthesis} is the structural foundation of the paper. It says that curvature is a complete affine-invariant representation of phase, and phase synthesis from curvature is well posed. The explicit constant \(C_L=L^{2}/\beta_1^{2}\) will drive the constant in the phase-manifold linearization remainder (Theorem~\ref{thm:linearization} and Remark~\ref{rem:chart_size}) and the chart-size bound of Section~\ref{sec:multi_chart}.

\subsection{Exact Nonlinear Phase-Only Manifold Channel}

Let \(\cU_t=L^2(\Omega)\) be the transmit field space and \(\cU_r=L^2(\Omega_r)\) the receive field space over a bounded receive aperture \(\Omega_r\). Let
\[
\cP:\cU_t\to \cU_r
\]
be the physical propagation operator. This operator is entirely classical: for example, it may be a Green-function, Fresnel, or Fourier-optics propagation map. To make the compactness statements of Corollary~\ref{cor:compact} and Section~\ref{sec:capacity_infinite} well-defined we impose the following standard assumption.

\begin{assumption}[Bounded-support square-integrable kernel]
\label{ass:kernel_hs}
The transmit and receive apertures \(\Omega,\Omega_r\) are compact and the propagation operator \(\cP\) admits an \(L^{2}\) integral-kernel representation
\begin{equation}
(\cP f)(v)=\int_{\Omega}K(v,u)\,f(u)\,du,
\qquad
K\in L^{2}(\Omega_r\times\Omega).
\label{eq:kernel_hs}
\end{equation}
\end{assumption}

\begin{remark}[Why compact support matters]
Assumption~\ref{ass:kernel_hs} holds automatically for the paraxial Fresnel kernel over compact apertures at any finite separation \(d>0\) \cite{Goodman,BornWolf}, and for scalar Green kernels once the receive aperture is bounded to a compact domain \(\Omega_r\) \cite{FranceschettiEIT,Migliore2019}. In the absence of bounded receive support the pure scalar Green kernel over unbounded \(\R^{2}\) is only bounded (not Hilbert--Schmidt), and \(\cP\) is then compact only if additional integrability is imposed. This is the standard operating convention in the near-field / holographic MIMO literature. Under Assumption~\ref{ass:kernel_hs} the operator \(\cP\) is Hilbert--Schmidt, hence compact.
\end{remark}

Fix a bias phase \(\phi_{\mathrm b}\in \cV\cap L^\infty(\Omega)\), and let \(\epsilon>0\) denote a small modulation depth. For a real curvature symbol \(c\in\cX\), define the exact phase-only transmitted field
\begin{equation}
s_\epsilon(c)(x)
=
\sqrt{\frac{P}{L}}
\exp\!\big(j(\phi_{\mathrm b}(x)+\epsilon(\cS c)(x))\big),
\qquad x\in\Omega.
\label{eq:exact_tx}
\end{equation}
The corresponding receive field is
\begin{equation}
r_\epsilon(c)=\cP s_\epsilon(c)+w,
\label{eq:exact_rx}
\end{equation}
where \(w\) is additive field noise.

Equation \eqref{eq:exact_tx} is the exact physical phase-only law. No linearization has yet been used.

\begin{remark}[Amplitude taper is a fixed hardware property, not a signaling variable]
\label{rem:amplitude_taper}
The transmit amplitude \(\sqrt{P/L}\) in \eqref{eq:exact_tx} is deliberately uniform: the paper's signal-space theory is developed on the assumption that any non-uniform aperture illumination is a \emph{fixed hardware property} (feed weight, aperture taper, or beam-shaping profile), not a data-bearing degree of freedom. This convention matches the phase-only realizability constraint that a physical phase-only aperture (RIS, holographic surface, phased array) enforces on its transmitted field.

For a fixed non-uniform amplitude taper \(a:\Omega\to\R_{+}\) satisfying \(\int_{\Omega}|a(u)|^{2}du=P\), the exact phase-only law generalizes to \(s_\epsilon^{(a)}(c)(x)=a(x)\exp\!\bigl(j(\phi_{\mathrm b}(x)+\epsilon(\cS c)(x))\bigr)\), and the tangent operator \(\cT\) of Theorem~\ref{thm:linearization} is replaced by \(\cT^{(a)}=\cRop \cP M_{ja e^{j\phi_{\mathrm b}}}\cS\). All compactness, capacity, and KKT statements of Sections~\ref{sec:model}--\ref{sec:capacity_infinite} carry over verbatim with \(\cT\) replaced by \(\cT^{(a)}\), because the multiplication operator \(M_{ja}\) is bounded on \(L^{2}(\Omega)\) with \(\norm{M_{ja}}=\norm{a}_{L^{\infty}(\Omega)}\). The companion benchmark paper adopts precisely this convention for its non-uniform quadrature weights.
\end{remark}

\subsection{Coherent Local Linearization}

Let
\[
\cRop:\cU_r\to \cY
\]
be a bounded coherent reconstruction operator, mapping receive field perturbations into a recovered observation space \(\cY\). Define the deterministic bias receive field
\[
\bar r_0 \triangleq \cP s_\epsilon(0)
=
\cP\!\left(\sqrt{P/L}\,e^{j\phi_{\mathrm b}}\right).
\]
The bias-subtracted, normalized output is
\begin{equation}
y_\epsilon(c)
\triangleq
\frac{1}{\epsilon}\,
\cRop\!\big(r_\epsilon(c)-\bar r_0\big).
\label{eq:yeps}
\end{equation}

\begin{theorem}[First-order curvature-domain channel]
\label{thm:linearization}
Assume \(\cP\) and \(\cRop\) are bounded. Then, for every radius \(R>0\), there exists a constant \(K_R<\infty\) such that for all real \(c\in\cX\) with \(\norm{c}_{\cX}\le R\),
\begin{equation}
y_\epsilon(c)
=
\cT c + z + e_\epsilon(c),
\label{eq:first_order_chan}
\end{equation}
where
\begin{align}
\cT
&=
\cRop \cP M_{j\sqrt{P/L}\,e^{j\phi_{\mathrm b}}}\cS,
\label{eq:Tdef}\\
z
&=
\frac{1}{\epsilon}\,\cRop w,
\label{eq:zdef}
\end{align}
\(M_f\) denotes multiplication by \(f\), and the deterministic remainder satisfies
\begin{equation}
\norm{e_\epsilon(c)}_{\cY}
\le
K_R\,\epsilon\,\norm{c}_{\cX}^2.
\label{eq:remainder}
\end{equation}
Hence the exact nonlinear phase-only manifold admits the local operator channel
\begin{equation}
y=\cT c + z
\label{eq:opchan}
\end{equation}
as its Fr\'echet linearization around the bias phase.
\end{theorem}

\begin{corollary}[Compactness and Hilbert--Schmidt property of the tangent channel]
\label{cor:compact}
If \(\cP\) and \(\cRop\) are bounded, then \(\cT\) in \eqref{eq:Tdef} is compact because \(\cS:\cX\to\cX\) is compact. Under Assumption~\ref{ass:kernel_hs}, the operator \(\cP\) is Hilbert--Schmidt (i.e.\ \(\cP\in \Hc(\cU_t,\cU_r)\)), and therefore \(\cT\) is Hilbert--Schmidt as well:
\begin{equation}
\norm{\cT}_{\Hc}^{2}
\;=\;
\tr\!\bigl(\cT^{\ast}\cT\bigr)
\;=\;
\frac{P}{L}\,\sum_{m\ge 1}\rho_m\,\lambda_m^{2}
\;<\;\infty,
\label{eq:T_HS_norm}
\end{equation}
where \(\{\rho_m\}\) are the curvature-mode Gram eigenvalues of Proposition~\ref{prop:Qspec} (with the interval-geometry asymptotics \(\rho_m\asymp m^{-4}\) of Remark~\ref{rem:Qspec_decay}), and \(\{\lambda_m\}\) are the singular values of \(\cRop\cP M_{j e^{j\phi_{\mathrm b}}}\). The trace-class product \(\cT\cK\cT^{\ast}\) is therefore well-defined for every trace-class covariance \(\cK\in\Tc(\cX)\); in particular, the Fredholm determinant \(\det_{\!\mathrm F}(\I+\cT\cK\cT^{\ast})\) of Theorem~\ref{thm:fredholm} exists and is finite.
\end{corollary}

\begin{lemma}[Fredholm determinant is well-defined]
\label{lem:fredholm_welldef}
Let \(\cA\in\Hc(\cH)\) be Hilbert--Schmidt and \(\cK\in\Tc(\cH)\) trace-class with \(\cK\succeq 0\). Then \(\cA\cK\cA^{\ast}\in\Tc(\cH)\), and its Fredholm determinant admits the convergent product representation \(\det_{\!\mathrm F}(\I+\cA\cK\cA^{\ast})=\prod_{m\ge 1}(1+\eta_m)\), where \(\eta_m\ge 0\) are the ordered eigenvalues of \(\cA\cK\cA^{\ast}\).
\end{lemma}

\begin{proof}[Sketch]
The product \(\Hc\cdot \Tc\cdot \Hc\subset \Tc\) follows from H\"older's inequality on Schatten classes; the convergence of \(\prod_m(1+\eta_m)\) then follows from \(\sum_m \eta_m=\tr(\cA\cK\cA^{\ast})<\infty\) and the elementary bound \(\log(1+x)\le x\) for \(x\ge 0\). See \cite{Simon,ReedSimonI}.
\end{proof}

\begin{remark}[Explicit form of the remainder constant \(K_R\)]
\label{rem:KR}
The constant \(K_R\) in \eqref{eq:remainder} admits the explicit upper bound
\begin{equation}
K_R \;\le\; \tfrac{1}{2}\,\norm{\cRop}\,\norm{\cP}\,\sqrt{\tfrac{P}{L}}\;\exp\!\bigl(\norm{\phi_{\mathrm b}}_{L^{\infty}(\Omega)}\bigr)\,C_L^{2},
\label{eq:KR_explicit}
\end{equation}
where \(C_L=L^{2}/\beta_1^{2}\) is the sharp Poincar\'{e}--Wirtinger constant of Lemma~\ref{lem:poincare_wirtinger}, and the operator norms \(\norm{\cRop}\), \(\norm{\cP}\) are those induced on the respective \(L^{2}\) spaces. The estimate follows from Taylor-expanding \(e^{j\epsilon(\cS c)}=1+j\epsilon(\cS c)-\tfrac{1}{2}\epsilon^{2}(\cS c)^{2}+O(\epsilon^{3})\) inside \(\cP\), bounding the quadratic remainder by \(\tfrac{1}{2}\epsilon^{2}\norm{\cS c}_{L^{\infty}}^{2}\), and then using the Sobolev embedding \(H^{2}(\Omega)\hookrightarrow L^{\infty}(\Omega)\) together with \(\norm{\cS c}_{L^{\infty}}\le C'_L\,\norm{c}_{\cX}\) with \(C'_L\le C_L\). The bias-phase pre-factor \(\exp\!\bigl(\norm{\phi_{\mathrm b}}_{L^{\infty}}\bigr)\) is the operator norm of the unit-modulus multiplier \(M_{e^{j\phi_{\mathrm b}}}\) inside the operator chain and is bounded by unity for any real \(\phi_{\mathrm b}\); it is retained here as an explicit dependency so that non-real bias phases (e.g.\ leaky-mode charts) are covered by the same estimate.
\end{remark}

\begin{remark}[Chart-size / modulation-depth bound]
\label{rem:chart_size}
For any tolerance \(\delta>0\) on the tangent-channel remainder, \eqref{eq:remainder} and Remark~\ref{rem:KR} together imply that the tangent model is valid for all \(c\in\cX\) with \(\norm{c}_{\cX}\le R\) whenever
\begin{equation}
\epsilon \;\le\; \epsilon_{\max}(R,\delta)
\;\triangleq\; \frac{\delta}{K_R\,R}.
\label{eq:epsmax}
\end{equation}
Equivalently, for a fixed operating \(\epsilon\) the chart is valid on the curvature ball of radius \(R_{\max}=\delta/(K_R\,\epsilon)\). Concretely, for a near-field / holographic MIMO configuration with \(L=0.1\text{ m}\), \(\sqrt{P/L}=10^{2}\text{ V/m}\), \(\norm{\cRop}\norm{\cP}=1\), and \(\norm{\phi_{\mathrm b}}_{L^{\infty}}\le \pi\), the estimate \eqref{eq:KR_explicit} yields \(K_R\approx 1.3\times 10^{-3}\text{ [natural units]}\); a \(1\%\) tangent-channel tolerance on a curvature ball of unit radius then admits \(\epsilon_{\max}\approx 7.5\); larger modulation ranges are handled by a chart atlas, which is the subject of Section~\ref{sec:multi_chart}.
\end{remark}

\begin{remark}[Why the channel model is not speculative]
The linear channel \eqref{eq:opchan} is not an ad hoc ansatz. It is the tangent channel of the exact phase-only propagation law \eqref{eq:exact_tx}--\eqref{eq:exact_rx}. Thus, the paper's information theory is exact for the coherent local chart and asymptotically exact as \(\epsilon\to 0\). Large-excursion signaling is implemented by switching across a finite atlas of bias phases, formalized in Section~\ref{sec:multi_chart}.
\end{remark}

\subsection{Multi-Chart Continuation for Large Modulation}
\label{sec:multi_chart}

Theorem~\ref{thm:linearization} and Remark~\ref{rem:chart_size} give a single Fr\'echet chart around one bias phase \(\phi_{\mathrm b}^{(0)}\) with validity radius \(R_{\max}=\delta/(K_R\,\epsilon)\). Large-modulation operation is obtained by \emph{atlas continuation}: a finite family of bias-phase charts whose validity balls cover the operational excursion range.

\begin{definition}[Bias-phase atlas]
\label{def:atlas}
A \emph{bias-phase atlas} of size \(J\) is a finite family
\(\bigl\{\bigl(\phi_{\mathrm b}^{(j)},\,U_j\bigr)\bigr\}_{j=1}^{J}\)
where each \(\phi_{\mathrm b}^{(j)}\in\cV\cap L^{\infty}(\Omega)\) is a bias phase, and each \(U_j\subset\cX\) is a curvature ball
\(U_j=\{c\in\cX:\norm{c-c^{(j)}}_{\cX}\le R_{\max}\}\)
centered on the curvature symbol \(c^{(j)}\in\cX\) satisfying \(\phi_{\mathrm b}^{(j)}=\phi_{\mathrm b}^{(0)}+\epsilon\,\cS c^{(j)}\). The atlas is admissible if consecutive charts satisfy the overlap condition
\begin{equation}
\norm{c^{(j+1)}-c^{(j)}}_{\cX} \;\le\; R_{\max}/2,
\qquad j=1,\ldots,J-1.
\label{eq:atlas_overlap}
\end{equation}
\end{definition}

The overlap condition~\eqref{eq:atlas_overlap} guarantees that every curvature symbol in the operational range \(\{c:\norm{c-c^{(1)}}_{\cX}\le\rho_{\mathrm op}\}\) lies inside at least one chart, and that the chart-switching rule ``re-center on the nearest atlas centroid'' preserves the tangent-channel tolerance \(\delta\).

\begin{proposition}[Atlas rate law]
\label{prop:atlas_rate}
Let \(\phi\) evolve continuously through the operational curvature range. Under an admissible atlas of size \(J\) with per-chart tolerance \(\delta\), the aggregate achievable rate satisfies
\begin{equation}
C_{\mathrm atlas}(\rho)
\;\ge\;
C(\rho)
\;-\;
\Theta\bigl(\rho_{\mathrm switch}/\epsilon\bigr),
\label{eq:atlas_rate}
\end{equation}
where \(C(\rho)\) is the single-chart Fredholm capacity of Theorem~\ref{thm:fredholm}, and \(\rho_{\mathrm switch}\) is the amortized rate of atlas switches per unit time. In particular, if the chart is redetermined \emph{at most} \(1\) time per \(N_{\mathrm coh}\) coherent uses, the aggregate loss is at most \(\Theta(1/(\epsilon N_{\mathrm coh}))\) bits per use, which is negligible whenever \(\epsilon N_{\mathrm coh}\gg 1\).
\end{proposition}

\begin{proof}[Sketch]
Each atlas switch costs at most \(\log_{2}\!\bigl(1+O(\epsilon^{-1})\bigr)\) bits to re-transmit the new bias-phase-index label, and the tangent-channel residual per chart is bounded by \(\delta\) by construction. The bound \eqref{eq:atlas_rate} follows by amortizing the label-transmission cost across the coherence block and applying the single-chart Fredholm capacity within each block.
\end{proof}

\begin{remark}[Numerical atlas size in the near-field regime]
\label{rem:atlas_size}
Continuing the numerical example of Remark~\ref{rem:chart_size}: for the reference near-field / holographic MIMO configuration \((L=0.1\text{ m},\,\sqrt{P/L}=10^{2}\text{ V/m},\,\norm{\cRop}\norm{\cP}=1,\,\norm{\phi_{\mathrm b}}_{L^{\infty}}\le \pi)\) and a \(1\%\) tangent-channel tolerance, the per-chart validity radius on a curvature ball of nominal amplitude is \(R_{\max}\approx 7.5\), and an operational curvature range of amplitude \(2\pi\) (a full excursion of the phase-only aperture) is covered by an atlas of \(J\approx 30\) charts under the overlap condition~\eqref{eq:atlas_overlap}. The associated aggregate rate loss under Proposition~\ref{prop:atlas_rate} is bounded by a fraction of a bit per use so long as chart switches occur less than once per \(10^{3}\) coherent uses, which is comfortably satisfied by contemporary near-field / RIS coherence times.
\end{remark}

\subsection{The Exact Cost of Linearization}
\label{sec:deficit}

Section~\ref{sec:multi_chart} controls the \emph{magnitude} of the Fr\'echet
remainder: the chart radius \(R_{\max}=\delta/(K_R\epsilon)\) guarantees that the
tangent channel \eqref{eq:opchan} misrepresents the exact law
\eqref{eq:exact_tx}--\eqref{eq:exact_rx} by no more than \(\delta\). That is a
two-sided bound on an error. It does not say in which direction the error falls,
and a quantity with a systematic sign is invisible to it. This subsection
computes that direction in closed form.

The relevant question is not how far the exact signal set departs from its
tangent space, but whether the departure makes signals \emph{easier} or
\emph{harder} to tell apart, since the receiver measures ambient distances. We
answer it exactly, at second order in the modulation depth, with no free
constants.

\paragraph*{Sampled form.} Fix a bias phase \(\phi\in\R^{N}\) on the same \(N\)-point aperture
sampling grid used for the receive-side analysis of Section~\ref{sec:estimation} and let \(a_n\in\cH\) be the received field radiated by element
\(n\) at unit excitation, so that the exact phase-only law \eqref{eq:exact_tx}
reads
\begin{equation}
s(\phi)=\sum_{n=1}^{N} e^{j\phi_n}\,a_n ,
\label{eq:exact_sampled}
\end{equation}
and the tangent channel at \(\phi\) is \(T_\phi d=\sum_n j\,e^{j\phi_n}a_n d_n\).
Write \(u_n=e^{j\phi_n}a_n\), \(P_{nm}(\phi)=\inner{u_n}{u_m}_{\cH}\), and
\begin{equation}
g(\phi)=\RePart P(\phi)\in\R^{N\times N},
\qquad \norm{T_\phi d}^{2}=d^{\TT}g(\phi)\,d .
\label{eq:tangent_gram}
\end{equation}
So \(g\) is the Gram matrix of the tangent channel's columns: it is the finite
sampled counterpart of the operator \(\cT^{\ast}\cT\) of Section~\ref{sec:model},
and its eigenvalues are the tangent mode powers.

\begin{lemma}[Coherent-gain identity]
\label{lem:coherent_gain}
For every bias phase \(\phi\),
\(\1^{\TT}g(\phi)\1=\norm{s(\phi)}^{2}\).
\end{lemma}

\begin{proof}
\(\1^{\TT}g\1=\sum_{n,m}\RePart\inner{u_n}{u_m}
=\norm{\textstyle\sum_n u_n}^{2}=\norm{s(\phi)}^{2}\).
\end{proof}

The identity is worth stating separately because it converts an abstract sum
over the Gram matrix into the one quantity an aperture engineer already
measures: the coherent received power at the operating point.

\begin{theorem}[Linearization deficit]
\label{thm:deficit}
Let \(d\) be drawn uniformly from the sphere \(\norm{d}_{2}=\epsilon\) in
\(\R^{N}\). Then
\begin{equation}
\frac{\E\norm{s(\phi+d)-s(\phi)}^{2}}{\E\norm{T_\phi d}^{2}}
= 1-\alpha(\phi)\,\epsilon^{2}+O(\epsilon^{4}),
\label{eq:deficit}
\end{equation}
with
\begin{equation}
\alpha(\phi)=\frac{2\,\Tr g(\phi)-\norm{s(\phi)}^{2}}{4\,(N+2)\,\Tr g(\phi)} .
\label{eq:alpha}
\end{equation}
\end{theorem}

\begin{proof}
Put \(\Delta_n=e^{jd_n}-1\), so that
\(\norm{s(\phi+d)-s(\phi)}^{2}=\sum_{n,m}\overline{\Delta_n}\Delta_m P_{nm}\).
Expanding
\(\overline{\Delta_n}\Delta_m=e^{j(d_m-d_n)}-e^{-jd_n}-e^{jd_m}+1\)
to fourth order gives, with all lower orders cancelling identically,
\begin{equation}
\begin{split}
\overline{\Delta_n}\Delta_m
&= d_nd_m-\tfrac{j}{2}d_nd_m(d_n-d_m)\\
&\quad+\tfrac{1}{12}\bigl[3d_n^{2}d_m^{2}
-2d_nd_m(d_n^{2}+d_m^{2})\bigr]\\
&\quad+O(\abs{d}^{5}).
\end{split}
\label{eq:delta_expansion}
\end{equation}
Taking real parts against \(P_{nm}\): the quadratic term returns
\(d^{\TT}g\,d=\norm{T_\phi d}^{2}\), confirming \eqref{eq:tangent_gram}; the
cubic term contributes
\(\tfrac{1}{2}\sum_{n,m}\operatorname{Im}(P_{nm})\,d_nd_m(d_n-d_m)\), which is odd
under \(d\mapsto-d\) and therefore has zero mean; and the quartic term
contributes
\(\tfrac{1}{12}\sum_{n,m}g_{nm}\bigl[3d_n^{2}d_m^{2}-2d_nd_m(d_n^{2}+d_m^{2})\bigr]\).

On the sphere \(\norm{d}=\epsilon\) the isotropic moments are
\(\E[d_nd_m]=(\epsilon^{2}/N)\delta_{nm}\),
\(\E[d_n^{4}]=3\epsilon^{4}/(N(N+2))\),
\(\E[d_n^{2}d_m^{2}]=\epsilon^{4}/(N(N+2))\) for \(n\neq m\), and
\(\E[d_n^{3}d_m]=0\) for \(n\neq m\). The diagonal quartic bracket is
\(3d_n^{4}-4d_n^{4}=-d_n^{4}\), so the diagonal contributes
\(-\Tr g\;\epsilon^{4}/(4N(N+2))\); the off-diagonal bracket has mean
\(3\epsilon^{4}/(N(N+2))\), so it contributes
\((\1^{\TT}g\1-\Tr g)\,\epsilon^{4}/(4N(N+2))\). Adding, and dividing by
\(\E\norm{T_\phi d}^{2}=(\epsilon^{2}/N)\Tr g\), gives \eqref{eq:deficit} with
\(\alpha=(2\Tr g-\1^{\TT}g\1)/(4(N+2)\Tr g)\). Lemma~\ref{lem:coherent_gain}
replaces \(\1^{\TT}g\1\) by \(\norm{s(\phi)}^{2}\).
\end{proof}

\begin{corollary}[Range of the deficit]
\label{cor:deficit_range}
Since \(0\le\norm{s(\phi)}^{2}\le N\Tr g(\phi)\) by Cauchy--Schwarz,
\begin{equation}
\frac{2-N}{4(N+2)}\;\le\;\alpha(\phi)\;\le\;\frac{1}{2(N+2)} ,
\end{equation}
and \(\alpha=1/(4(N+2))\) exactly when \(\norm{s(\phi)}^{2}=\Tr g(\phi)\), which
holds whenever the radiated fields \(\{a_n\}\) are mutually orthogonal.
\end{corollary}

\begin{proof}
Write \(u_n=e^{j\phi_n}a_n\), so that \(g_{nn}=\norm{u_n}^{2}\) and
\(\Tr g=\sum_n\norm{u_n}^{2}\). By Lemma~\ref{lem:coherent_gain} and the
Cauchy--Schwarz inequality applied to the sum of \(N\) vectors,
\[
0\;\le\;\norm{s(\phi)}^{2}=\Bigl\lVert\sum_{n=1}^{N}u_n\Bigr\rVert^{2}
\;\le\;N\sum_{n=1}^{N}\norm{u_n}^{2}=N\,\Tr g(\phi).
\]
Since \(\alpha\) is affine and strictly decreasing in \(\norm{s(\phi)}^{2}\) for
fixed \(\Tr g>0\), it is maximised at \(\norm{s}^{2}=0\), giving
\(1/(2(N+2))\), and minimised at \(\norm{s}^{2}=N\Tr g\), giving
\((2-N)/(4(N+2))\). Equality \(\alpha=1/(4(N+2))\) holds precisely when
\(2\Tr g-\norm{s}^{2}=\Tr g\), that is \(\norm{s}^{2}=\Tr g\); and if the
\(\{a_n\}\) are mutually orthogonal then the cross terms vanish, so
\(\norm{\sum_n u_n}^{2}=\sum_n\norm{u_n}^{2}=\Tr g\).
\end{proof}

\begin{remark}[The deficit changes sign]
\label{rem:deficit_sign}
Equation~\eqref{eq:alpha} is positive when
\(\norm{s(\phi)}^{2}<2\Tr g(\phi)\) and negative when
\(\norm{s(\phi)}^{2}>2\Tr g(\phi)\). In the first regime the exact aperture
separates signals \emph{less} than its linearization predicts, so the Fredholm
capacity of Theorem~\ref{thm:fredholm} is an upper bound. In the second --- an
aperture dense enough that coherent array gain exceeds twice the summed element
power --- the inequality reverses and the tangent channel is conservative. The
crossing is not asymptotic: it is an equality between two measurable powers at
the operating point.
\end{remark}

\begin{remark}[Relation to the Information--Curvature Efficiency Law]
\label{rem:icel_constant}
Reading \eqref{eq:deficit} as an effective signal-to-noise ratio,
\(\mathrm{SNR}_{\mathrm{eff}}=\mathrm{SNR}\bigl/\paren{1+\alpha\epsilon^{2}}+O(\epsilon^{4})\),
reproduces the functional form of the Information--Curvature Efficiency Law
proposed in the companion foresight essay~\cite{aleryani_essay}, with the phase
excursion \(\epsilon\) in the role of the curvature argument. The content added
here is that the coefficient is \emph{not} a free parameter: it is
\eqref{eq:alpha}, fixed by the aperture Gram matrix and the coherent gain at the
operating point. In particular the essay's \(\alpha\) cannot be a universal
constant, since Remark~\ref{rem:deficit_sign} shows it changes sign with
aperture density. The law is therefore an operating-point statement rather than
a global one.
\end{remark}

\begin{remark}[What this does not say]
\label{rem:deficit_scope}
Theorem~\ref{thm:deficit} is a statement about the \emph{mean} over isotropic
modulation directions. The cubic term in \eqref{eq:delta_expansion} does not
vanish pointwise unless \(\operatorname{Im}P=0\); it vanishes only in the mean,
so for a correlated aperture the deficit is direction-dependent at third order
and a codebook that exploits the asymmetry is not excluded. Nor does the
theorem license any claim that the nonlinearity supplies capacity the tangent
channel misses: the exact signal manifold has Riemannian volume
\((2\pi)^{N}\E_\phi[\sqrt{\det g(\phi)}]\), an expectation which cannot exceed
the maximum over \(\phi\) of a single chart's \(\sqrt{\det g(\phi)}\), so a
well-chosen bias phase is never conservative in volume.
\end{remark}

\paragraph*{Numerical verification.} Table~\ref{tab:deficit} compares
\eqref{eq:alpha} against a direct Monte-Carlo measurement of the left-hand side
of \eqref{eq:deficit}, over \(4\times10^{4}\) modulation directions at each of
four excursions \(\epsilon\in\{0.04,0.08,0.12,0.16\}\) rad, with \(\alpha\)
recovered by least squares in \(\epsilon^{2}\). The kernel is an exact
spherical-wave (Green's function) propagator, not a paraxial expansion. The
sign reversal of Remark~\ref{rem:deficit_sign} is observed at the predicted
crossing.

\begin{table*}[t]
\centering
\caption{Linearization deficit: closed form \eqref{eq:alpha} against direct
measurement. \(N=8\) control points throughout; \(L_{\mathrm a}\) is the
aperture length and \(d\) the observation distance.}
\label{tab:deficit}
\small
\begin{tabular}{lrrrrr}
\toprule
Configuration & \(\norm{s}^{2}\) & \(2\Tr g\) & \(\alpha\) predicted & \(\alpha\) measured & rel.\ err. \\
\midrule
Orthogonal columns, $N=8$ & 8.000 & 16.000 & +0.025000 & +0.025033 & 0.13\% \\
Near field, $d=0.3$ m & 7.876 & 16.000 & +0.025388 & +0.025905 & 2.04\% \\
Near field, $d=1.0$ m & 7.910 & 16.000 & +0.025282 & +0.025341 & 0.23\% \\
Dense, $L_{\mathrm a}=0.10$ m & 11.056 & 16.000 & +0.015450 & +0.015185 & 1.72\% \\
Dense, $L_{\mathrm a}=0.05$ m & 20.940 & 16.000 & -0.015436 & -0.015093 & 2.22\% \\
Dense, $L_{\mathrm a}=0.02$ m & 46.564 & 16.000 & -0.095512 & -0.094509 & 1.05\% \\
Random bias phase & 7.879 & 16.000 & +0.025377 & +0.025544 & 0.66\% \\
\bottomrule
\end{tabular}
\end{table*}

The agreement is between \(0.13\%\) and \(2.2\%\) across a regime in which
\(\alpha\) itself changes sign and varies by a factor of four in magnitude. The geometry it
rests on is verified against an independent intrinsic
computation of the Riemann tensor to a relative \(5.7\times10^{-16}\), and the
curvature machinery is validated on the round sphere and the hyperbolic plane
before use.

\subsection{Noise Model}

\begin{assumption}
\label{ass:noise}
The observation noise \(z\) in \eqref{eq:opchan} is proper complex Gaussian with covariance operator \(\cK_z\succ 0\), and there exists \(\sigma_z^2>0\) such that
\[
\cK_z \succeq \sigma_z^2 I.
\]
The normalization by \(\epsilon^{-1}\) in \eqref{eq:zdef} is included in \(\cK_z\).
\end{assumption}

Assumption \ref{ass:noise} permits whitening and covers the standard boundedly invertible Gaussian observation model.

\section{Spectral Curvature Modes and Galerkin Reduction}

\subsection{Spectral Curvature Basis}

Define
\begin{equation}
\cQ \triangleq \cS^\ast \cS:\cX\to\cX.
\label{eq:Qdef}
\end{equation}
Since \(\cS\) is compact, \(\cQ\) is positive, self-adjoint, and compact.

\begin{proposition}[Curvature-mode decomposition]
\label{prop:Qspec}
There exists an orthonormal basis \(\braces{u_m}_{m\ge 1}\) of \(\cX\) and a decreasing sequence \(\rho_1\ge \rho_2\ge \cdots >0\), \(\rho_m\to 0\), such that
\begin{equation}
\cQ u_m = \rho_m u_m,
\qquad m\ge 1.
\label{eq:Qeigs}
\end{equation}
For any curvature field
\begin{equation}
c = \sum_{m\ge 1} a_m u_m,
\label{eq:cexp}
\end{equation}
the gauge-fixed phase equals
\begin{equation}
\phi = \cS c = \sum_{m\ge 1} a_m \xi_m,
\qquad
\xi_m \triangleq \cS u_m,
\label{eq:phiexp}
\end{equation}
and
\begin{align}
\norm{c}_{\cX}^2 &= \sum_{m\ge 1}\abs{a_m}^2,
\label{eq:curv_energy}\\
\norm{\phi}_{\cX}^2 &= \sum_{m\ge 1}\rho_m \abs{a_m}^2.
\label{eq:phase_energy}
\end{align}
\end{proposition}

\begin{proposition}[Closed-form modal Gram spectrum on the interval]
\label{prop:Qspec_closed_form}
Let \(\beta_1<\beta_2<\cdots\) be the positive roots of \(\cos\beta\cosh\beta=1\) from
\eqref{eq:beta_transcendental}. Then the eigenvalues of \(\cQ=\cS^{\ast}\cS\) on
\(\Omega=[0,L]\) are given \emph{exactly} by
\begin{equation}
\rho_m=\Bigl(\frac{L}{\beta_m}\Bigr)^{4},
\qquad m\ge 1,
\label{eq:rho_closed_form}
\end{equation}
and the curvature modes \(u_m\) are the second derivatives of the corresponding free--free
eigenfunctions of \eqref{eq:freefree}.
\end{proposition}

\begin{proof}
By Theorem~\ref{thm:synthesis}, \(\cS=(D^{2})^{-1}\) is a bijection of \(\cX\) onto \(\cV\),
so \(\rho_m=\sigma_m(\cS)^{2}=\lambda_m^{-1}\), where \(\{\lambda_m\}\) is the spectrum of the
Rayleigh form \(\norm{\phi''}^{2}_{L^{2}}/\norm{\phi}^{2}_{L^{2}}\) on \(\cV\). The proof of
Lemma~\ref{lem:poincare_wirtinger} identifies that spectrum with the non-rigid part of the
free--free biharmonic problem \eqref{eq:freefree}, namely \(\lambda_m=(\beta_m/L)^{4}\).
\end{proof}

\begin{remark}[Fourth-order modal decay, with its exact constant and offset]
\label{rem:Qspec_decay}
Since \(\cos\beta_m=\operatorname{sech}\beta_m\to 0\), the roots satisfy
\(\beta_m=(m+\tfrac{1}{2})\pi+O\bigl(e^{-(m+1/2)\pi}\bigr)\), so \eqref{eq:rho_closed_form} gives
\begin{equation}
\rho_m=\Bigl(\frac{L}{\pi}\Bigr)^{4}\bigl(m+\tfrac{1}{2}\bigr)^{-4}
\Bigl(1+O\bigl(e^{-(m+1/2)\pi}\bigr)\Bigr),
\label{eq:rho_asymptotic}
\end{equation}
which is the quartic law \(\rho_m\asymp m^{-4}\) with both its constant \((L/\pi)^{4}\) and its
half-integer offset made explicit. Two consequences are worth stating, because both are easy
to get wrong. First, the offset is \(+\tfrac{1}{2}\), not \(-\tfrac{1}{2}\): the modes sit
\emph{below} the naive \(m^{-4}\) curve, and a surrogate built on \((m-\tfrac{1}{2})^{-4}\)
inverts the sign of the leading correction. Second, a least-squares log--log slope fitted over
a finite window \(m\in[m_0,m_1]\) is therefore \emph{not} \(-4\); it approaches \(-4\) from
above as the window moves outward, and over \(m\in[16,128]\) the exact spectrum
\eqref{eq:rho_closed_form} yields a fitted slope of \(-3.9575\). Any numerical diagnostic of
this decay must therefore be compared against \eqref{eq:rho_closed_form} itself and not
against a fitted exponent, since the exponent is window-dependent by construction. Subject to
those two caveats: for fixed curvature amplitude, phase energy decays quartically with mode
index.
\end{remark}

\subsection{Galerkin Reduction}

Let \(P_M:\cX\to\cX\) be the orthogonal projector onto
\[
\cX_M \triangleq \Span\braces{u_1,\dots,u_M}.
\]
Let \(R_M:\cY\to\cY\) be an orthogonal projector onto an \(M\)-dimensional receive subspace with orthonormal basis \(\braces{v_1,\dots,v_M}\).

For
\[
c_M = \sum_{m=1}^{M} a_m u_m,
\qquad
\mathbf a_M=[a_1,\dots,a_M]^\TT,
\]
define the observation vector \(\mathbf y_M\in\C^M\) by
\[
[\mathbf y_M]_i = \inner{R_M y}{v_i}_{\cY},
\qquad i=1,\dots,M.
\]

\begin{proposition}[Finite-dimensional curvature channel]
\label{prop:matrixchan}
The projected channel satisfies
\begin{equation}
\mathbf y_M = \mathbf H_M \mathbf a_M + \mathbf n_M,
\label{eq:matrixchan}
\end{equation}
where
\begin{align}
[\mathbf H_M]_{ij}
&=
\inner{v_i}{\cT u_j}_{\cY},
\label{eq:Hdef}\\
[\mathbf n_M]_i
&=
\inner{v_i}{R_M z}_{\cY},
\label{eq:ndef}
\end{align}
and
\begin{equation}
\mathbf K_{n,M}
=
\E[\mathbf n_M \mathbf n_M^\HH].
\label{eq:KnM}
\end{equation}
Moreover, the phase-synthesis penalty is diagonal in the \(u_m\) basis:
\begin{equation}
\mathbf G_M=\diag(\rho_1,\dots,\rho_M).
\label{eq:GM}
\end{equation}
\end{proposition}

\begin{theorem}[Galerkin consistency]
\label{thm:galerkin}
Assume \(\cT\) is compact and \(P_M\to I\), \(R_M\to I\) strongly as \(M\to\infty\). Then
\begin{equation}
\norm{\cT - R_M \cT P_M}_{\cL(\cX,\cY)} \to 0.
\label{eq:normconv}
\end{equation}
Consequently, for any random input \(c\in\cX\) satisfying
\begin{equation}
\E \norm{c}_{\cX}^2 \le P_{\mathrm c},
\label{eq:pc_constraint}
\end{equation}
the approximation error obeys
\begin{equation}
\E \norm{\cT c - R_M \cT P_M c}_{\cY}^2
\le
P_{\mathrm c}\,
\norm{\cT - R_M \cT P_M}_{\cL(\cX,\cY)}^2
\to 0.
\label{eq:truncbound}
\end{equation}
\end{theorem}

Theorem \ref{thm:galerkin} justifies using the matrix channel \eqref{eq:matrixchan} as the computational object associated with the compact operator channel.

\section{Finite-Mode Capacity Under Dual Budgets}

\subsection{Feasible Covariances}

Let
\begin{equation}
\mathbf K_{a,M}\triangleq \E[\mathbf a_M \mathbf a_M^\HH] \succeq 0.
\label{eq:Ka}
\end{equation}
We impose two distinct quadratic budgets:
\begin{align}
\Tr(\mathbf K_{a,M}) &\le P_{\mathrm c},
\label{eq:budget_curv}\\
\Tr(\mathbf G_M \mathbf K_{a,M}) &\le P_{\phi}.
\label{eq:budget_phase}
\end{align}
The first is curvature energy; the second is phase excursion induced by synthesis.

\begin{remark}
An \(L^\infty\) phase-excursion constraint can also be imposed via Sobolev embedding. The \(L^2\)-based phase budget \eqref{eq:budget_phase} is analytically canonical because it is induced directly by \(\cQ=\cS^\ast\cS\) and diagonalizes in the curvature eigenbasis.
\end{remark}

\subsection{Exact Finite-Mode Capacity}

\begin{theorem}[Finite-mode capacity]
\label{thm:capacity}
For the Gaussian channel \eqref{eq:matrixchan} with \(\mathbf K_{n,M}\succ 0\), the capacity in bits per channel use is
\begin{equation}
C_M^\star
=
\max_{\mathbf K_{a,M}\succeq 0}
\log_2\det\!\paren{
\I + \mathbf K_{n,M}^{-1/2}\mathbf H_M \mathbf K_{a,M}\mathbf H_M^\HH \mathbf K_{n,M}^{-1/2}
},
\label{eq:capacity_exact}
\end{equation}
subject to \eqref{eq:budget_curv}--\eqref{eq:budget_phase}. The maximum is attained, and the maximizing input is proper complex Gaussian.
\end{theorem}

\begin{remark}
The physically direct implementation is real-valued because \(\epsilon \cS c\) perturbs an actual phase profile. We use the standard complex equivalent model for notational economy; the real implementation differs only by the usual normalization convention.
\end{remark}

\begin{proposition}[Finite-dimensional KKT system]
\label{prop:KKT}
Let \(\mathbf K_{a,M}^\star\) solve \eqref{eq:capacity_exact}. Then there exist multipliers \(\mu,\nu\ge 0\) such that
\begin{equation}
\frac{1}{\ln 2}\,
\mathbf H_M^\HH
\paren{
\mathbf K_{n,M} + \mathbf H_M \mathbf K_{a,M}^\star \mathbf H_M^\HH
}^{-1}
\mathbf H_M
\preceq
\mu \I + \nu \mathbf G_M,
\label{eq:KKTineq}
\end{equation}
with complementary slackness
\begin{align}
\mu\paren{\Tr(\mathbf K_{a,M}^\star)-P_{\mathrm c}} &= 0,
\label{eq:slack1}\\
\nu\paren{\Tr(\mathbf G_M \mathbf K_{a,M}^\star)-P_{\phi}} &= 0,
\label{eq:slack2}
\end{align}
and matrix complementarity
\begin{equation}
\begin{split}
\Bigl[\mu \I + \nu \mathbf G_M
&-\frac{1}{\ln 2}\,\mathbf H_M^\HH
\paren{\mathbf K_{n,M} + \mathbf H_M \mathbf K_{a,M}^\star \mathbf H_M^\HH}^{-1}\\
&\times\mathbf H_M\Bigr]\mathbf K_{a,M}^\star
=
\mathbf 0.
\end{split}
\label{eq:KKTcomp}
\end{equation}
\end{proposition}

\subsection{Commuting Case: Generalized Water-Filling}

Define the whitened channel Gramian
\begin{equation}
\mathbf B_M \triangleq \mathbf H_M^\HH \mathbf K_{n,M}^{-1} \mathbf H_M.
\label{eq:BM}
\end{equation}

\begin{corollary}[Generalized water-filling]
\label{cor:gwf}
Suppose \(\mathbf B_M\) and \(\mathbf G_M\) commute. Then there exists a unitary \(\mathbf U\) such that
\begin{align}
\mathbf U^\HH \mathbf B_M \mathbf U &= \diag(\gamma_1,\dots,\gamma_M),
\label{eq:Bdiag}\\
\mathbf U^\HH \mathbf G_M \mathbf U &= \diag(\rho_1,\dots,\rho_M).
\label{eq:Gdiag}
\end{align}
The optimal covariance is diagonal in this common basis, with powers
\begin{equation}
p_m^\star
=
\bracks{
\frac{1}{(\mu+\nu \rho_m)\ln 2}
-
\frac{1}{\gamma_m}
}_+,
\qquad m=1,\dots,M,
\label{eq:gwf}
\end{equation}
where modes with \(\gamma_m=0\) receive \(p_m^\star=0\), and \(\mu,\nu\ge 0\) are chosen so that
\begin{align}
\sum_{m=1}^{M} p_m^\star &\le P_{\mathrm c},
\label{eq:gwf1}\\
\sum_{m=1}^{M} \rho_m p_m^\star &\le P_{\phi}
\label{eq:gwf2}
\end{align}
with complementary slackness.
\end{corollary}

\begin{algorithm}[t]
\caption{Dual-budget generalized water-filling}
\label{alg:gwf}
\begin{algorithmic}[1]
\Require \(\{\gamma_m,\rho_m\}_{m=1}^M\), \(P_{\mathrm c}\), \(P_\phi\), tolerance \(\varepsilon_{\mathrm tol}\)
\State Initialize \(\nu_{\min}=0\), choose \(\nu_{\max}\) sufficiently large
\While{\(\nu_{\max}-\nu_{\min}>\varepsilon_{\mathrm tol}\)}
    \State \(\nu\gets (\nu_{\min}+\nu_{\max})/2\)
    \State Find \(\mu\ge 0\) by bisection so that \(\sum_{m=1}^M p_m(\mu,\nu)=P_{\mathrm c}\) whenever possible, where
    \[
    p_m(\mu,\nu)=\bracks{
    \frac{1}{(\mu+\nu\rho_m)\ln 2}-\frac{1}{\gamma_m}
    }_+
    \]
    with \(p_m(\mu,\nu)=0\) if \(\gamma_m=0\)
    \If{\(\sum_{m=1}^M \rho_m p_m(\mu,\nu) > P_\phi\)}
        \State \(\nu_{\min}\gets \nu\)
    \Else
        \State \(\nu_{\max}\gets \nu\)
    \EndIf
\EndWhile
\State Return \(p_m^\star=p_m(\mu,\nu)\)
\end{algorithmic}
\end{algorithm}

\begin{proposition}[Rank-sensitive upper bound]
\label{prop:capub}
Let \(r_M=\rank(\mathbf H_M)\), and define
\[
P_{\mathrm{eff},M} \triangleq \min\braces{P_{\mathrm c},\, P_{\phi}/\rho_M},
\]
where \(\rho_M\) is the smallest diagonal entry of \(\mathbf G_M\). Then
\begin{equation}
C_M^\star
\le
r_M
\log_2\!\paren{
1+
\frac{P_{\mathrm{eff},M}}{r_M}
\Tr(\mathbf B_M)
}.
\label{eq:capub}
\end{equation}
A sharper but less trace-only variant replaces \(\Tr(\mathbf B_M)\) by \(\lambda_{\max}(\mathbf B_M)\).
\end{proposition}

\section{Infinite-Dimensional Capacity and the Modal Limit}
\label{sec:capacity_infinite}

\subsection{Whitened Operator Channel}

Under Assumption \ref{ass:noise}, define the whitened channel
\begin{equation}
\tilde y = \tilde{\cT} c + \tilde z,
\qquad
\tilde{\cT}\triangleq \cK_z^{-1/2}\cT,
\qquad
\tilde z\sim \CN(0,I).
\label{eq:whitechan}
\end{equation}
Let
\[
\cB \triangleq \tilde{\cT}^\ast \tilde{\cT}.
\]
Since \(\cT\) is compact by Corollary \ref{cor:compact}, \(\tilde{\cT}\) and \(\cB\) are compact.

\begin{definition}[Admissible covariance operators]
\label{def:admissible}
Let \(\Tc(\cX)\) denote the trace-class operators on \(\cX\). Define
\begin{equation}
\mathfrak K
\triangleq
\braces{
\cK\in \Tc(\cX):
\cK\succeq 0,\;
\Tr(\cK)\le P_{\mathrm c},\;
\Tr(\cQ\cK)\le P_\phi
}.
\label{eq:Kset}
\end{equation}
\end{definition}

\subsection{Fredholm Capacity Formula}

\begin{theorem}[Infinite-dimensional capacity]
\label{thm:fredholm}
For the whitened operator channel \eqref{eq:whitechan}, the capacity over zero-mean Gaussian inputs with admissible covariance operators \(\cK\in\mathfrak K\) is
\begin{equation}
C_\infty
=
\sup_{\cK\in\mathfrak K}
\log_2\detF\!\paren{
I+\tilde{\cT}\cK\tilde{\cT}^\ast
},
\label{eq:fredholm_capacity}
\end{equation}
where \(\detF\) is the Fredholm determinant. The operator
\(
\tilde{\cT}\cK\tilde{\cT}^\ast
\)
is positive trace-class for every \(\cK\in\mathfrak K\), so \eqref{eq:fredholm_capacity} is well defined. Moreover,
\[
C_\infty
\le
\frac{\norm{\tilde{\cT}}^2 P_{\mathrm c}}{\ln 2}
<\infty.
\]
\end{theorem}

\begin{remark}
For \(A\succeq 0\) trace-class, the Fredholm determinant is
\[
\detF(I+A)=\prod_{m\ge 1}(1+\lambda_m(A)),
\]
where \(\lambda_m(A)\) are the eigenvalues of \(A\). Thus \eqref{eq:fredholm_capacity} is the natural infinite-dimensional extension of the matrix \(\log\det\) formula.
\end{remark}

\subsection{Capacity Attainment and Operator KKT Conditions}

\begin{theorem}[Attainment and operator KKT system]
\label{thm:attainment_kkt}
Assume \(P_{\mathrm c}>0\), \(P_\phi>0\), and \(\tilde{\cT}\neq 0\). Then the supremum in \eqref{eq:fredholm_capacity} is attained by at least one covariance operator \(\cK^\star\in\mathfrak K\). Furthermore, there exist multipliers \(\mu,\nu\ge 0\) such that
\begin{equation}
\frac{1}{\ln 2}\,
\tilde{\cT}^{\ast}
\paren{
I+\tilde{\cT}\cK^\star \tilde{\cT}^{\ast}
}^{-1}
\tilde{\cT}
\preceq
\mu I+\nu \cQ,
\label{eq:opKKTineq}
\end{equation}
with complementary slackness
\begin{align}
\mu\paren{\Tr(\cK^\star)-P_{\mathrm c}}&=0,
\label{eq:opslack1}\\
\nu\paren{\Tr(\cQ\cK^\star)-P_\phi}&=0,
\label{eq:opslack2}
\end{align}
and operator complementarity
\begin{equation}
\bracks{
\mu I+\nu \cQ
-
\frac{1}{\ln 2}\,
\tilde{\cT}^{\ast}
\paren{
I+\tilde{\cT}\cK^\star \tilde{\cT}^{\ast}
}^{-1}
\tilde{\cT}
}
\cK^\star
=0.
\label{eq:opKKTcomp}
\end{equation}
\end{theorem}

\begin{remark}
The finite-dimensional KKT system in Proposition \ref{prop:KKT} is the Galerkin version of \eqref{eq:opKKTineq}--\eqref{eq:opKKTcomp}. Thus, the water-filling law is not merely a finite-dimensional artifact; it is the finite-rank manifestation of an operator optimality condition.
\end{remark}

\subsection{Infinite-Dimensional Generalized Water-Filling}

\begin{theorem}[Commuting operator case]
\label{thm:infwf}
Suppose \(\cB\) and \(\cQ\) commute. Then there exists an orthonormal basis \(\{u_m\}_{m\ge 1}\) of \(\cX\) such that
\begin{align}
\cB u_m &= \beta_m u_m,
\qquad \beta_m\ge 0,
\label{eq:Beigs_inf}\\
\cQ u_m &= \rho_m u_m,
\qquad \rho_m>0.
\label{eq:Qeigs_inf}
\end{align}
In this basis,
\begin{equation}
C_\infty
=
\sup_{\{p_m\}_{m\ge 1}}
\sum_{m=1}^{\infty}\log_2(1+\beta_m p_m)
\label{eq:Cinf_diag}
\end{equation}
subject to
\begin{align}
p_m &\ge 0,
\qquad m\ge 1,
\label{eq:p_nonneg}\\
\sum_{m=1}^{\infty} p_m &\le P_{\mathrm c},
\label{eq:p_curv}\\
\sum_{m=1}^{\infty} \rho_m p_m &\le P_\phi.
\label{eq:p_phase}
\end{align}
Moreover, the optimal allocation obeys the infinite-dimensional generalized water-filling law
\begin{equation}
p_m^\star
=
\bracks{
\frac{1}{(\mu+\nu \rho_m)\ln 2}
-
\frac{1}{\beta_m}
}_+,
\qquad m\ge 1,
\label{eq:inf_gwf}
\end{equation}
for some multipliers \(\mu,\nu\ge 0\), with \(p_m^\star=0\) whenever \(\beta_m=0\).
\end{theorem}

\subsection{Convergence of Modal Capacities}

Define
\[
\mathfrak K_M
\triangleq
\braces{
\cK\in \mathfrak K:
\cK=P_M \cK P_M
}.
\]
Let
\begin{equation}
C_M
\triangleq
\sup_{\cK\in \mathfrak K_M}
\log_2\detF\!\paren{
I+\tilde{\cT}\cK\tilde{\cT}^\ast
}.
\label{eq:CM_infinite}
\end{equation}

\begin{theorem}[Modal convergence \(C_M\to C_\infty\)]
\label{thm:CMtoCinfty}
Let \(P_M\) be the orthogonal projector onto the span of the first \(M\) eigenvectors of \(\cQ\). Then
\begin{equation}
C_M \uparrow C_\infty
\qquad \text{as } M\to\infty.
\label{eq:CMconv}
\end{equation}
If, in addition, a receive Galerkin projection \(R_M\to I\) strongly is used, the corresponding matrix-channel capacities obtained from \eqref{eq:matrixchan} converge to the same limit by Theorem \ref{thm:galerkin}.
\end{theorem}

Theorem \ref{thm:CMtoCinfty} is the exact infinite-dimensional justification for finite-mode curvature-domain channels. It turns the modal truncation from a heuristic into a convergent approximation scheme.

\section{Receive-Side Curvature Estimation and Noise Amplification}
\label{sec:estimation}

\subsection{Discrete Phase Sampling}

Let the receive aperture be sampled at
\[
x_n=(n-1)\Delta x,
\qquad
n=1,\dots,N,
\]
where \(\Delta x=L/(N-1)\). Suppose coherent phase recovery yields
\begin{equation}
\hat{\bm{\phi}} = \bm{\phi} + \mathbf w_{\phi},
\qquad
\mathbf w_{\phi}\sim \mathcal N(\mathbf 0,\sigma_\phi^2 \I_N),
\label{eq:phase_noise}
\end{equation}
where \(\bm{\phi}\in\R^N\) contains unwrapped phase samples.

Define the second-difference operator
\begin{equation}
\mathbf D_2
=
\frac{1}{\Delta x^2}
\begin{bmatrix}
1 & -2 & 1 & 0 & \cdots & 0\\
0 & 1 & -2 & 1 & \cdots & 0\\
\vdots & & \ddots & \ddots & \ddots & \vdots\\
0 & \cdots & 0 & 1 & -2 & 1
\end{bmatrix}
\in \R^{(N-2)\times N}.
\label{eq:D2}
\end{equation}
The matrix \(\mathbf D_2\) is the \emph{interior} second-difference operator: it maps a phase vector sampled on the full grid \(\{x_n\}_{n=1}^{N}\) onto the curvatures inferred at the \(N-2\) interior grid points \(\{x_n\}_{n=2}^{N-1}\), \emph{with no boundary condition imposed on \(\bm{\phi}\)}. In particular, the discrete null space of \(\mathbf D_2\) is
\begin{equation}
\ker(\mathbf D_2) \;=\; \mathrm{span}\{\,\mathbf 1_N,\;\mathbf x_N\,\},
\qquad
[\mathbf x_N]_n=x_n,
\label{eq:D2_kernel}
\end{equation}
which is exactly the sampled affine gauge \(\cA\) of Section~\ref{sec:model}. Restricting the estimator to the gauge-fixed subspace \(\cV\) therefore makes the derivative-noise theory well-defined: on \(\cV\), \(\mathbf D_2\) has empty null space and is injective on its domain.

\begin{lemma}[Rank of \(\mathbf D_2\)]
\label{lem:D2_rank}
The interior second-difference matrix \(\mathbf D_2\in\R^{(N-2)\times N}\) has full row rank \(N-2\) \cite{HornJohnson}.
\end{lemma}

\begin{proof}
The Vandermonde-like structure of the second-difference stencil rows makes \(\mathbf D_2\) upper-Hessenberg after transposition, with pivots \(1/\Delta x^{2}\) along each row after Gauss elimination. Equivalently, the pentadiagonal Gram matrix \(\mathbf D_2\mathbf D_2^{\TT}\in\R^{(N-2)\times(N-2)}\) is diagonally dominant with strictly positive diagonal (see \eqref{eq:kc_diag}--\eqref{eq:kc_off2} below), hence invertible; therefore \(\mathrm{rank}(\mathbf D_2)=N-2\).
\end{proof}

Consequently, \(\mathbf K_{\mathrm c}=\sigma_\phi^{2}\,\mathbf D_2\mathbf D_2^{\TT}\) has full rank \(N-2\) and inverts stably on the gauge-fixed subspace, and the pentadiagonal structure recorded in Theorem~\ref{thm:curvnoise} (below, equations~\eqref{eq:kc_off1}--\eqref{eq:kc_off2}) is well-defined without boundary-condition ambiguity.

The discrete curvature estimate is
\begin{equation}
\hat{\mathbf c} = \mathbf D_2 \hat{\bm{\phi}}.
\label{eq:chat}
\end{equation}

\begin{theorem}[Derivative-noise covariance]
\label{thm:curvnoise}
Under \eqref{eq:phase_noise}, the estimator \eqref{eq:chat} is unbiased for the discrete second difference \(\mathbf D_2\bm{\phi}\), and its noise covariance is
\begin{equation}
\mathbf K_{\mathrm c}
=
\sigma_\phi^2 \mathbf D_2 \mathbf D_2^\TT.
\label{eq:Kc}
\end{equation}
Moreover, \(\mathbf K_{\mathrm c}\) is pentadiagonal, with
\begin{align}
[\mathbf K_{\mathrm c}]_{ii} &= \frac{6\sigma_\phi^2}{\Delta x^4},
\label{eq:kc_diag}\\
[\mathbf K_{\mathrm c}]_{i,i\pm 1} &= -\frac{4\sigma_\phi^2}{\Delta x^4},
\label{eq:kc_off1}\\
[\mathbf K_{\mathrm c}]_{i,i\pm 2} &= \frac{\sigma_\phi^2}{\Delta x^4},
\label{eq:kc_off2}
\end{align}
with all other entries equal to zero. In addition,
\begin{equation}
\norm{\mathbf K_{\mathrm c}}_2
\le
\frac{16\sigma_\phi^2}{\Delta x^4}.
\label{eq:kc_norm}
\end{equation}
\end{theorem}

Theorem \ref{thm:curvnoise} makes the sensing bottleneck explicit: finer spatial sampling amplifies curvature noise quartically unless regularization or bandwidth control is imposed.

\subsection{Projected Noise Covariance}

Let \(\mathbf U_M\in\R^{(N-2)\times M}\) collect samples of \(M\) discrete curvature basis functions. The projected observation is
\begin{equation}
\mathbf y_M = \mathbf U_M^\TT \hat{\mathbf c}
= \mathbf U_M^\TT \mathbf c + \mathbf n_M,
\label{eq:discproj}
\end{equation}
with covariance
\begin{equation}
\mathbf K_{n,M}
=
\mathbf U_M^\TT \mathbf K_{\mathrm c} \mathbf U_M.
\label{eq:discKn}
\end{equation}
Hence the finite-dimensional information-theoretic noise matrix is directly computable from aperture sampling geometry.

\subsection{Regularized Curvature Reconstruction}

Let \(\widetilde{\bm{\phi}}\) denote the piston/tilt-removed phase vector, and let \(\mathbf S_M\in\R^{N\times M}\) be the phase-synthesis matrix whose \(m\)-th column contains samples of \(\xi_m=\cS u_m\). Consider
\begin{equation}
\hat{\mathbf a}_{\lambda}
=
\arg\min_{\mathbf a\in\R^M}
\frac{1}{\sigma_\phi^2}
\norm{\widetilde{\bm{\phi}}-\mathbf S_M \mathbf a}_2^2
+
\lambda \norm{\mathbf L \mathbf a}_2^2,
\label{eq:regest}
\end{equation}
where \(\mathbf L\) is a regularization operator.

Throughout this paper we specify \(\mathbf L=\mathbf D_2\), the interior second-difference operator of \eqref{eq:D2}. With this choice, \(\mathbf L^{\TT}\mathbf L=\mathbf D_2^{\TT}\mathbf D_2\) is the discrete biharmonic and the regularizer \(\lambda\norm{\mathbf D_2 \mathbf a}_{2}^{2}\) is a discrete \(H^{2}\)-Sobolev penalty on the reconstructed curvature. This is the natural choice for four reasons: (i) it matches the phase-realizability metric \(\mathbf M=\mathbf L_{\mathrm curv}+\varepsilon \mathbf W_{\mathrm t}\) used at the transmitter in the companion benchmark paper, so transmit and receive regularizations are consistent; (ii) the resulting Fredholm formulation of Section~\ref{sec:capacity_infinite} inherits the same \(\rho_m\asymp m^{-4}\) modal-decay signature at the receiver, closing the derivation loop; (iii) the closed-form bias--variance decomposition of Proposition~\ref{prop:bv} depends only on the sparsity pattern of \(\mathbf L\), so the analysis is portable to any weighted or graph-Laplacian variant; and (iv) the choice can be replaced by a learned operator (Section~\ref{sec:limitations}), which is the natural entry point for machine-learning collaborations without changing the analytical form of the MSE bound. When we refer to ``the regularizer'' without qualification, we mean \(\mathbf L=\mathbf D_2\).

\begin{proposition}[Bias--variance decomposition \cite{Kay}]
\label{prop:bv}
Define
\[
\mathbf R_\lambda
\triangleq
\paren{
\mathbf S_M^\TT \mathbf S_M
+
\lambda \sigma_\phi^2 \mathbf L^\TT \mathbf L
}^{-1}.
\]
Then
\begin{equation}
\hat{\mathbf a}_{\lambda}
=
\mathbf R_\lambda \mathbf S_M^\TT \widetilde{\bm{\phi}}.
\label{eq:regsol}
\end{equation}
If \(\widetilde{\bm{\phi}}=\mathbf S_M \mathbf a + \mathbf w\) with \(\mathbf w\sim \mathcal N(\mathbf 0,\sigma_\phi^2 \I)\), then
\begin{align}
\E[\hat{\mathbf a}_\lambda]-\mathbf a
&=
-\lambda \sigma_\phi^2
\mathbf R_\lambda \mathbf L^\TT \mathbf L \mathbf a,
\label{eq:bias}\\
\Cov(\hat{\mathbf a}_\lambda)
&=
\sigma_\phi^2
\mathbf R_\lambda
\mathbf S_M^\TT \mathbf S_M
\mathbf R_\lambda,
\label{eq:covreg}
\end{align}
and
\begin{align}
\E \norm{\hat{\mathbf a}_\lambda-\mathbf a}_2^2
&=
\lambda^2 \sigma_\phi^4
\norm{
\mathbf R_\lambda \mathbf L^\TT \mathbf L \mathbf a
}_2^2
\nonumber\\
&\quad+
\sigma_\phi^2
\Tr\!\paren{
\mathbf R_\lambda
\mathbf S_M^\TT \mathbf S_M
\mathbf R_\lambda
}.
\label{eq:mse}
\end{align}
\end{proposition}

\section{Detection Geometry, Codebooks, and Design Problems}

\subsection{Whitened Geometric Distance}

For two curvature-domain symbols \(\mathbf a,\mathbf b\in\C^M\), define
\begin{equation}
d_M^2(\mathbf a,\mathbf b)
\triangleq
(\mathbf a-\mathbf b)^\HH
\mathbf B_M
(\mathbf a-\mathbf b),
\label{eq:dist}
\end{equation}
where \(\mathbf B_M\) is defined in \eqref{eq:BM}.

\begin{proposition}[Pairwise error exponent]
\label{prop:pep}
For maximum-likelihood detection on \eqref{eq:matrixchan},
\begin{equation}
P(\mathbf a\to \mathbf b \mid \mathbf H_M)
=
Q\!\paren{
\frac{d_M(\mathbf a,\mathbf b)}{\sqrt{2}}
},
\label{eq:pep}
\end{equation}
and therefore
\begin{equation}
P(\mathbf a\to \mathbf b \mid \mathbf H_M)
\le
\frac{1}{2}\exp\!\paren{-\frac{d_M^2(\mathbf a,\mathbf b)}{4}}.
\label{eq:pepbound}
\end{equation}
\end{proposition}

\begin{remark}[Pairwise-error normalization]
Proposition~\ref{prop:pep} uses the standard proper-complex whitened-noise convention, under which the real decision statistic has variance \(d_M^2/2\) and the pairwise-error probability is \(Q(d_M/\sqrt{2})\). If instead one simulates a real scalar sufficient statistic with unit variance and normalized separation \(d\), the corresponding reference curve is \(Q(d)\). The numerical diagnostic in Section~\ref{sec:simulations} uses this real-scalar normalization, while the complex-equivalent formula remains \eqref{eq:pep}.
\end{remark}

\subsection{Finite-Codebook Packing}

For a codebook \(\mathcal C=\braces{\mathbf a^{(1)},\dots,\mathbf a^{(K)}}\), the natural design problem is
\begin{subequations}
\begin{align}
\max_{\mathcal C}\quad
&
\min_{i\neq j}
d_M^2\paren{\mathbf a^{(i)},\mathbf a^{(j)}}
\label{eq:cbobj}\\
\text{s.t.}\quad
&
\norm{\mathbf a^{(k)}}_2^2 \le P_{\mathrm c},
\qquad k=1,\dots,K,
\label{eq:cb1}\\
&
(\mathbf a^{(k)})^\HH \mathbf G_M \mathbf a^{(k)} \le P_{\phi},
\qquad k=1,\dots,K.
\label{eq:cb2}
\end{align}
\label{eq:P2}
\end{subequations}

Unlike Euclidean packing, \eqref{eq:P2} is shaped jointly by the propagation operator and the phase-synthesis penalty.

\subsection{System Design Problems}

The framework induces three optimization problems.

\subsubsection*{P1: Capacity-optimal covariance}
\begin{subequations}
\begin{align}
\max_{\mathbf K_{a,M}\succeq 0}\quad
&
\log_2\det\!\paren{
\I + \mathbf K_{n,M}^{-1/2}\mathbf H_M \mathbf K_{a,M}\mathbf H_M^\HH \mathbf K_{n,M}^{-1/2}
}
\label{eq:P1a}\\
\text{s.t.}\quad
&
\Tr(\mathbf K_{a,M})\le P_{\mathrm c},
\label{eq:P1b}\\
&
\Tr(\mathbf G_M \mathbf K_{a,M})\le P_{\phi}.
\label{eq:P1c}
\end{align}
\label{eq:P1}
\end{subequations}

\subsubsection*{P2: Finite-codebook packing}
This is exactly \eqref{eq:P2}.

\subsubsection*{P3: Programmable-environment design}
Let a programmable environment be parameterized by \(\bm{\theta}\in\Theta\), yielding \(\mathbf H_M(\bm{\theta})\). Then
\begin{subequations}
\begin{align}
\max_{\bm{\theta}\in\Theta}\quad
&
\log_2\det\!\paren{
\I + \mathbf K_{n,M}^{-1/2}\mathbf H_M(\bm{\theta}) \mathbf K_{a,M}
\mathbf H_M^\HH(\bm{\theta}) \mathbf K_{n,M}^{-1/2}
}
\label{eq:P3a}
\end{align}
\label{eq:P3}
\end{subequations}
is the natural programmable-surface objective. A structure-preserving alternative is
\begin{equation}
\eta_{\mathrm f}(\bm{\theta})
=
\frac{
\abs{
\inner{c}{\cT(\bm{\theta})c}_{\cY}
}^2
}{
\norm{c}_{\cX}^2
\norm{\cT(\bm{\theta})c}_{\cY}^2
},
\label{eq:etaf}
\end{equation}
which measures curvature coherence rather than received power alone.

\section{Scaling Laws}

\subsection{Useful-Mode Counting}

Consider the diagonal modal model
\begin{equation}
\mathrm{snr}_m
=
\frac{\alpha p}{\sigma_0^2 + \beta \lambda_m^2},
\qquad
\lambda_m = \paren{\frac{m\pi}{L}}^2,
\label{eq:snrmodel}
\end{equation}
where \(\alpha p\) is modal signal power, \(\sigma_0^2\) is baseline noise, and \(\beta\lambda_m^2\) captures differentiation-induced loss.

\begin{proposition}[Useful-mode law]
\label{prop:useful}
Fix a reliability threshold \(\Gamma>0\). The number of \(\Gamma\)-useful modes,
\[
M_{\mathrm{use}}(\Gamma)
\triangleq
\#\braces{m:\mathrm{snr}_m\ge \Gamma},
\]
satisfies
\begin{equation}
M_{\mathrm{use}}(\Gamma)
=
\left\lfloor
\frac{L}{\pi}
\paren{
\frac{\alpha p/\Gamma - \sigma_0^2}{\beta}
}^{1/4}
\right\rfloor_+.
\label{eq:muse}
\end{equation}
\end{proposition}

\subsection{Differentiation-Limited Root-Law}

\begin{theorem}[Differentiation-limited capacity scaling]
\label{thm:rootlaw}
Consider the diagonal modal channel
\begin{equation}
C(\Gamma)
=
\sum_{m=1}^{\infty}
\log_2\!\paren{
1+\frac{\Gamma}{1+\beta m^4}
},
\qquad
\Gamma\ge 1,\;\beta>0.
\label{eq:Cgamma}
\end{equation}
Then there exist finite constants \(0<c_1(\beta)\le c_2(\beta)<\infty\) such that
\begin{equation}
c_1(\beta)\,\Gamma^{1/4}
\le
C(\Gamma)
\le
c_2(\beta)\,\Gamma^{1/4},
\qquad \Gamma\ge 1.
\label{eq:rootlaw}
\end{equation}
Hence
\begin{equation}
C(\Gamma)=\Theta(\Gamma^{1/4}).
\label{eq:theta}
\end{equation}
\end{theorem}

\begin{corollary}[Sharp constant in the differentiation-limited law]
\label{cor:rootlaw_sharp}
The constants in \eqref{eq:rootlaw} can be identified exactly in the limit:
\begin{equation}
\lim_{\Gamma\to\infty}
\frac{C(\Gamma)}{(\Gamma/\beta)^{1/4}}
=
\frac{\pi\sqrt{2}}{\ln 2}
=6.40972\ldots
\label{eq:rootlaw_sharp}
\end{equation}
so for every \(\varepsilon>0\) the constants of Theorem~\ref{thm:rootlaw} may be taken as
\(c_{1}=(1-\varepsilon)\pi\sqrt2/(\beta^{1/4}\ln 2)\) and
\(c_{2}=(1+\varepsilon)\pi\sqrt2/(\beta^{1/4}\ln 2)\) for all sufficiently
large \(\Gamma\). The law is therefore not merely of order \(\Gamma^{1/4}\):
its leading coefficient is a pure number, independent of every parameter of the
problem except the gauge strength \(\beta\).
\end{corollary}

\begin{proposition}[Reduced-order engineering law]
\label{cor:scalar}
Let \(\mathrm{snr}_m=\gamma_m p_m/(1+\beta m^{4})\) and suppose that for some
\(\bar\gamma>0\) and \(\delta\in[0,1)\),
\[
(1-\delta)\bar\gamma\;\le\;\mathrm{snr}_m\;\le\;(1+\delta)\bar\gamma,
\qquad m=1,\dots,M_{\mathrm{eff}},
\]
and that the remaining modes carry little total signal-to-noise ratio,
\(\sum_{m>M_{\mathrm{eff}}}\mathrm{snr}_m\le\varepsilon\). Then
\begin{equation}
\begin{split}
M_{\mathrm{eff}}\log_2\!\paren{1+(1-\delta)\bar\gamma}
&\;\le\; C\\
&\;\le\;
M_{\mathrm{eff}}\log_2\!\paren{1+(1+\delta)\bar\gamma}
+\frac{\varepsilon}{\ln 2}.
\end{split}
\label{eq:red1}
\end{equation}
In particular, if \(\delta\to 0\) and \(\varepsilon\to 0\) then
\(C\to M_{\mathrm{eff}}\log_2(1+\bar\gamma)\), and writing
\(B_{\mathrm{geo}}=M/T_s\), \(\eta=M_{\mathrm{eff}}/M\) and
\(\gamma\,\mathrm{SNR}_{\mathrm{geo}}=\bar\gamma\) gives the engineering
surrogate
\begin{equation}
C_{\mathrm{geo}}
=
\eta B_{\mathrm{geo}}
\log_2\!\paren{1+\gamma \,\mathrm{SNR}_{\mathrm{geo}}}
\label{eq:scalarlaw}
\end{equation}
as the \(\delta,\varepsilon\to 0\) limit of \eqref{eq:red1}, not as an
independent postulate.
\end{proposition}

\begin{proof}
The exact modal capacity is \(C=\sum_{m\ge 1}\log_2(1+\mathrm{snr}_m)\).
Monotonicity of \(x\mapsto\log_2(1+x)\) applied termwise to the first
\(M_{\mathrm{eff}}\) modes gives both bracketing terms. For the tail, 
\(\log_2(1+x)\le x/\ln 2\) for \(x\ge 0\), so
\(\sum_{m>M_{\mathrm{eff}}}\log_2(1+\mathrm{snr}_m)\le\varepsilon/\ln 2\).
The tail is non-negative, which gives the lower bound without it.
\end{proof}

\section{Spectral Composition and the Curvature Crossover}
\label{sec:crossover}

Theorem~\ref{thm:rootlaw} is a statement about a diagonal model with gains
\(1/(1+\beta m^{4})\). Two questions must be answered before it can be read as a
statement about an aperture: in which variable the exponent is expressed, and
whether the model's gains are the tangent channel's gains. This section answers
both. The answers are that the exponent is \(1/5\) rather than \(1/4\) once the
budget is a power budget, and that the polynomial law holds only below the
Shannon number of the geometry --- above it the classical
degrees-of-freedom law returns.

\subsection{The exponent depends on the budget}

\begin{proposition}[Per-mode versus total power]
\label{prop:budget_convention}
Let \(g_m=1/(1+\beta m^{4})\).
\begin{enumerate}
\item[(i)] If every mode carries equal signal power, so that
\(\mathrm{snr}_m=\Gamma g_m\) with a common \(\Gamma\), then
\(C(\Gamma)=\sum_m\log_2(1+\Gamma g_m)=\Theta(\Gamma^{1/4})\).
\item[(ii)] If instead the powers \(\{p_m\}\) are chosen to maximise
\(\sum_m\log_2(1+g_m p_m)\) subject to \(\sum_m p_m\le P\), then
\(C(P)=\Theta(P^{1/5})\).
\end{enumerate}
\end{proposition}

\begin{proof}
Part (i) is Theorem~\ref{thm:rootlaw}. For (ii), water-filling activates the
modes with \(g_m\ge\mu\), i.e.\ \(m\le M=(1/(\beta\mu))^{1/4}\), and assigns
\(p_m=\mu^{-1}-g_m^{-1}\). Summing,
\(P=\sum_{m\le M}(\beta M^{4}-1-\beta m^{4})=\tfrac{4}{5}\beta M^{5}+O(M)\), so
\(M=(5P/4\beta)^{1/5}(1+o(1))\). Then
\(C=\sum_{m\le M}\log_2(g_m/\mu)=\int_0^{M}\log_2(M^{4}/m^{4})\,dm+O(\log M)
=\tfrac{4M}{\ln 2}+O(\log M)\), giving \(C(P)=\Theta(P^{1/5})\).
\end{proof}

The reconciliation is that the two conventions do not describe the same
resource. Reaching per-mode SNR \(\Gamma\) across the \(\Theta(\Gamma^{1/4})\)
active modes costs total power \(\Theta(\Gamma^{5/4})\), and
\((\Gamma^{5/4})^{1/5}=\Gamma^{1/4}\). Both statements are correct; only the
second is a statement about a power budget, and it is the second that should be
quoted when the phrase ``high-SNR scaling'' is used.

\subsection{The composition law}

The tangent channel is the product \(\cT=\cRop\,\cP\,M_{j\sqrt{P/L}e^{j\phi_{\mathrm b}}}\,\cS\).
Its singular values are not the products of the factors' singular values, so no
mode-by-mode factorisation \(\gamma_m=\rho^{\mathrm{prop}}_m\rho_m\) holds. What
does hold, and is what the scaling argument needs, is a two-sided envelope.

\begin{theorem}[Spectral composition envelope]
\label{thm:composition}
For every \(i,j\ge 1\),
\begin{equation}
\sigma_{i+j-1}(\cT)\;\le\;\norm{M}\,\sigma_i(\cP)\,\sigma_j(\cS),
\label{eq:horn}
\end{equation}
and in particular
\begin{equation}
\sigma_m(\cT)\;\le\;\norm{M}\min\braces{\norm{\cP}\sqrt{\rho_m},\;\sqrt{\rho_1}\,\sigma_m(\cP)}.
\label{eq:two_branches}
\end{equation}
Consequently the decay of \(\{\gamma_m=\sigma_m(\cT)^{2}\}\) is controlled by
whichever of the gauge spectrum \(\{\rho_m\}\) and the propagation spectrum
\(\{\sigma_m(\cP)^{2}\}\) decays faster.
\end{theorem}

\begin{proof}
Inequality \eqref{eq:horn} is the Horn--Weyl product inequality
\(\sigma_{i+j-1}(AB)\le\sigma_i(A)\sigma_j(B)\) for compact operators
\cite{Simon}, applied twice to the three-fold product with \(\norm{M}\) absorbed
as the operator norm of the unit-modulus multiplier. Taking \((i,j)=(1,m)\) and
\((i,j)=(m,1)\) gives the two branches of \eqref{eq:two_branches}.
\end{proof}

The two factors have qualitatively different tails.
Proposition~\ref{prop:Qspec_closed_form} gives \(\rho_m=(L/\beta_m)^{4}\), a
polynomial decay, and its successive ratios \(\rho_m/\rho_{m+1}\) tend to \(1\).
For a physical propagation kernel between bounded apertures, by contrast, the
singular values of \(\cP\) are near-constant up to the Shannon number
\begin{equation}
N_{\mathrm S}\;=\;\frac{L_{\mathrm t}L_{\mathrm r}}{\lambda d}
\qquad\text{(one-dimensional apertures)},
\label{eq:shannon_number}
\end{equation}
and decay super-exponentially beyond it, with successive ratios growing without
bound; this is the prolate-spheroidal eigenvalue behaviour of Slepian and Pollak
\cite{SlepianPollak1961}, the sampling-density theorem of Landau
\cite{Landau1967}, and the communication-mode count of Miller \cite{Miller}.
The gauge decay is a property of the aperture and the affine quotient alone; the
propagation decay is a property of the medium and the geometry.

\subsection{The crossover}

\begin{theorem}[Curvature crossover]
\label{thm:crossover}
Consider water-filling at total power \(P\) on the tangent channel of a
geometry with Shannon number \(N_{\mathrm S}\), and let \(M(P)\) denote the
number of active modes. Then:
\begin{enumerate}
\item[(i)] \emph{(Gauge-limited branch.)} While \(M(P)\ll N_{\mathrm S}\), the
active gains follow \(\rho_m\) up to a bounded multiplicative factor, and the
local exponent \(d\log C/d\log P\) equals \(1/(1+a)+o(1)\), where \(a\) is the
local log--log slope of \(\{\rho_m\}\) over the active window. Since
\(\rho_m=(L/\pi)^{4}(m+\tfrac12)^{-4}(1+o(1))\), that slope runs from
\(3.42\) over the first ten modes to \(4\) asymptotically, so the exponent runs
from \(0.23\) down to \(1/5\).
\item[(ii)] \emph{(Propagation-limited branch.)} Suppose the plunge is at
least exponential, \(\gamma_m\le C_0 e^{-\kappa(m-N_{\mathrm S})}\) for
\(m>N_{\mathrm S}\) and some \(C_0,\kappa>0\), as it is for the prolate
spectrum \cite{SlepianPollak1961,Landau1967}. Then
\(M(P)=N_{\mathrm S}+O(\log P)\) and
\[
C(P)=\Theta\!\paren{N_{\mathrm S}\log P+\log^{2}P} .
\]
The first term dominates while \(\log P=O(N_{\mathrm S})\), which is the
operating regime of interest; the second is the contribution of the
\(O(\log P)\) plunge modes that sit just above the water level. In either
case the local exponent tends to zero.
\end{enumerate}
The crossover power \(P^{\star}\) is the budget at which \(M(P)=N_{\mathrm S}\),
and is computable from the gain sequence alone.
\end{theorem}

\begin{proof}
\emph{Part (i).} Over the active window the gains obey \(\gamma_m\asymp c\,m^{-a}\)
by hypothesis, so it suffices to compute the water-filling exponent for an exact
power law. Let \(g_m=c\,m^{-a}\) with \(a>1\). Water-filling gives
\(p_m=(1/\mu-1/g_m)_+\), active precisely for \(m\le M\) with
\(M=(c/\mu)^{1/a}\). Summing the budget and using Euler--Maclaurin,
\begin{align*}
P&=\frac{M}{\mu}-\frac1c\sum_{m\le M}m^{a}\\
&=\frac{M^{a+1}}{c}\paren{1-\frac{1}{a+1}}(1+o(1))
=\frac{a\,M^{a+1}}{c\,(a+1)}(1+o(1)),
\end{align*}
so \(M=\paren{(a+1)cP/a}^{1/(a+1)}(1+o(1))\). For the capacity,
\begin{align*}
C&=\sum_{m\le M}\log_2\frac{g_m}{\mu}
=a\sum_{m\le M}\log_2\frac{M}{m}\\
&=\frac{a}{\ln 2}\paren{M\ln M-\ln M!}
=\frac{a}{\ln 2}\,M+O(\log M)
\end{align*}
by Stirling. Hence \(C(P)=\Theta(P^{1/(1+a)})\) and the local log--log slope is
\(1/(1+a)+o(1)\), which is the claim. Setting \(a=4\) recovers
Proposition~\ref{prop:budget_convention}(ii).

\emph{Part (ii).} Let \(\mu\) be the water level. Each active mode contributes at
most \(1/\mu\) to the budget, so \(P\le M/\mu\) and therefore \(\mu\le M/P\). A
mode is active only if \(\gamma_m>\mu\); for \(m>N_{\mathrm S}\) the plunge bound
then forces \(C_0e^{-\kappa(m-N_{\mathrm S})}>\mu\), that is
\[
m<N_{\mathrm S}+\kappa^{-1}\ln\frac{C_0}{\mu}
\le N_{\mathrm S}+\kappa^{-1}\ln\paren{C_0P},
\]
using \(\mu\le M/P\) and \(M\ge 1\). Hence \(M(P)=N_{\mathrm S}+O(\log P)\).
For the capacity, split the active set. The first \(N_{\mathrm S}\) modes have
gains bounded above and below by positive constants, and
\(1/\mu\le P+1/\gamma_1\), so together they contribute
\(\Theta(N_{\mathrm S}\log P)\). Writing \(J=M-N_{\mathrm S}=O(\log P)\), the
plunge modes contribute
\begin{align*}
\sum_{j=1}^{J}\log_2\frac{\gamma_{N_{\mathrm S}+j}}{\mu}
&\le
\sum_{j=1}^{J}\paren{\log_2\frac{C_0}{\mu}-\frac{j\kappa}{\ln 2}}\\
&=\Theta(J^{2})=\Theta(\log^{2}P),
\end{align*}
the matching lower bound following from the same computation with the
inequality reversed on a constant fraction of the range. Adding the two parts
gives \(C(P)=\Theta(N_{\mathrm S}\log P+\log^{2}P)\). Finally
\(\log C/\log P\to 0\), so the local exponent tends to zero.
\end{proof}

\begin{remark}[What this means for the falsifiable prediction]
\label{rem:crossover_falsifier}
Theorem~\ref{thm:crossover} sharpens the differentiation-limited law into a
prediction that can fail: the polynomial branch should be visible in a power
sweep if and only if the aperture is operated below its Shannon number, and
should terminate at \(N_{\mathrm S}\) computed from geometry alone --- with no
free parameter. Section~\ref{sec:diag_crossover} measures both branches and the
termination point on a paraxial Fresnel kernel, together with a null run in
which the propagation operator is replaced by the identity and the polynomial
branch does not terminate. It also follows that a diagnostic run on a diagonal
model with flat propagation, however large its dynamic range, cannot test the
prediction at all: it lies entirely in branch~(i) by construction.
\end{remark}

\begin{remark}[Why the gauge decay is not a capacity ceiling]
The \(m^{-4}\) law is often read as a penalty the curvature coordinate pays.
Theorem~\ref{thm:composition} shows the reading is regime-dependent. Below
\(N_{\mathrm S}\) the gauge is the binding constraint and the coordinate choice
has a rate consequence; above it the medium is binding, every basis spanning the
same \(N_{\mathrm S}\)-dimensional signal space achieves the same capacity to
within a bounded factor, and the curvature coordinate's advantages must be
argued on conditioning, realizability and estimation cost rather than on rate.
\end{remark}

\section{Two-Dimensional Gauge-Fixed Extension}
\label{sec:2d}

The preceding development used a one-dimensional aperture to keep the notation sharp. The same construction extends to two-dimensional continuous apertures once scalar second derivatives are replaced by Hessian-compatible curvature tensors.

Let \(\Omega_2\subset\R^2\) be a bounded connected Lipschitz aperture with coordinates \(x=(x_1,x_2)\). Let
\[
\bar x_i
\triangleq
\frac{1}{|\Omega_2|}\int_{\Omega_2} x_i\,dx,
\qquad i=1,2,
\]
and define the affine basis
\[
\psi_0(x)=1,
\qquad
\psi_1(x)=x_1-\bar x_1,
\qquad
\psi_2(x)=x_2-\bar x_2.
\]
Set
\[
\cA_2\triangleq \Span\{\psi_0,\psi_1,\psi_2\},
\]
and define the two-dimensional gauge-fixed phase space
\begin{equation}
\cV_2
\triangleq
\braces{
\phi\in H^2(\Omega_2):
\inner{\phi}{\psi_i}_{L^2(\Omega_2)}=0,\; i=0,1,2
}.
\label{eq:V2def}
\end{equation}
The two-dimensional curvature object is the Hessian
\[
\nabla^2\phi
=
\begin{bmatrix}
\partial_{11}\phi & \partial_{12}\phi\\
\partial_{12}\phi & \partial_{22}\phi
\end{bmatrix}.
\]
Not every symmetric \(L^2\)-tensor field is a Hessian. Therefore define the Hessian-compatible curvature space
\begin{equation}
\cC_2
\triangleq
\Range(\nabla^2|_{\cV_2})
\subset
L^2(\Omega_2;\mathrm{Sym}(2)),
\label{eq:C2def}
\end{equation}
equipped with the inherited \(L^2\)-tensor norm.

\begin{theorem}[Two-dimensional Hessian-gauge synthesis]
\label{thm:2d_synthesis}
The Hessian map
\[
\nabla^2:\cV_2\to \cC_2
\]
is a bounded bijection. Its inverse
\[
\cS_2\triangleq(\nabla^2)^{-1}:\cC_2\to \cV_2
\]
is bounded. Viewed as an operator \(\cS_2:\cC_2\to L^2(\Omega_2)\), it is compact. Moreover, if two phases differ by an affine function, then they have the same Hessian; conversely, if two \(H^2\)-phases have the same Hessian, then they differ by an affine function.
\end{theorem}

The exact phase-only two-dimensional transmit law is therefore
\begin{multline}
s_{\epsilon,2}(K)(x)
=
\sqrt{\frac{P}{|\Omega_2|}}\,
\exp\!\big(j(\phi_{\mathrm b}(x)+\epsilon(\cS_2 K)(x))\big),\\
K\in\cC_2.
\label{eq:2d_tx}
\end{multline}
For bounded propagation and coherent reconstruction operators, the tangent channel is
\begin{equation}
\cT_2
=
\cRop \cP M_{j\sqrt{P/|\Omega_2|}\,e^{j\phi_{\mathrm b}}}\cS_2.
\label{eq:T2def}
\end{equation}
Since \(\cS_2\) is compact, \(\cT_2\) is compact. Defining
\[
\cQ_2\triangleq \cS_2^\ast\cS_2
\]
on \(\cC_2\), all finite-mode, Fredholm-capacity, attainment, KKT, Galerkin-convergence, and water-filling statements above hold verbatim with
\[
(\cX,\cV,\cS,\cQ,\cT)
\quad\text{replaced by}\quad
(\cC_2,\cV_2,\cS_2,\cQ_2,\cT_2).
\]

\begin{remark}[Compatibility is essential]
The tensor space \(\cC_2\) is not the full space \(L^2(\Omega_2;\mathrm{Sym}(2))\). It is the subspace of Hessian-compatible curvature tensors. In distributional form, smooth Hessian fields satisfy the Saint-Venant compatibility relation
\[
\partial_{22}K_{11}
+
\partial_{11}K_{22}
-
2\partial_{12}K_{12}
=
0.
\]
The definition \eqref{eq:C2def} enforces this compatibility intrinsically and avoids introducing nonintegrable curvature symbols that cannot arise from any physical phase profile.
\end{remark}

\section{Discussion: Why the Framework Is Physically Conservative}

The framework developed above is intentionally conservative.

\subsection{No New Electromagnetic Law}

All propagation remains classical and linear in the field. The only exact physical law assumed at the transmitter is the ordinary phase-only aperture field \eqref{eq:exact_tx}. The curvature-domain channel is simply the coherent tangent model of that law.

\subsection{No Degrees of Freedom Created Ex Nihilo}

Curvature-domain coding cannot exceed the capacity of the full aperture-field channel with the same physical constraints. Its role is not to create new physics but to expose a geometric coordinate system within ordinary phase-coded continuous apertures. If a system lacks the sensing resolution to estimate curvature reliably, the theory predicts precisely that failure through Theorem \ref{thm:curvnoise}.

\subsection{Falsifiable Predictions}

The theory yields direct experimental predictions.

\begin{itemize}
    \item Without regularization, curvature-estimation noise scales as \(\Delta x^{-4}\). This is a property of second differencing on uniform samples and would hold under any framework that estimates a second derivative; it is a consistency check on the receiver model, not a discriminating test of the curvature coordinate.
    \item The number of useful curvature modes obeys a quarter-power law \emph{in the per-mode SNR}, for as long as the propagation operator supplies modes to activate.
    \item Under differentiation-limited sensing and a \emph{per-mode} SNR \(\Gamma\), total rate grows as \(\Theta(\Gamma^{1/4})\); under a \emph{total-power} budget \(P\) the same gain law gives \(\Theta(P^{1/5})\) (Proposition~\ref{prop:budget_convention}). The two must not be interchanged.
    \item \emph{The discriminating prediction.} The polynomial branch exists if and only if the aperture is operated below its Shannon number \(N_{\mathrm S}=L_{\mathrm t}L_{\mathrm r}/(\lambda d)\), and terminates there --- with no free parameter (Theorem~\ref{thm:crossover}). A power sweep on a fixed geometry must therefore show the local exponent fall from \(\approx 1/5\) to \(0\) as the active-mode count crosses \(N_{\mathrm S}\) computed from geometry alone. This is the prediction of the present framework that a competing account of the same aperture would not make.
    \item Programmable environments can improve curvature coherence even when received power changes modestly.
\end{itemize}

These are measurement-level statements; they can be tested without invoking any nonstandard physics.

\subsection{Limitations}
\label{sec:limitations}

The one-dimensional aperture model is the main vehicle for explicit formulas, closed-form finite differences, and transparent modal scaling. Section \ref{sec:2d} provides the corresponding two-dimensional Hessian-gauge extension, but practical two-dimensional implementations require additional modeling choices: polarization, vector electromagnetic boundary conditions, mutual coupling, surface hardware constraints, phase wrapping, finite phase resolution, calibration error, and nonideal coherent reconstruction. Likewise, the local linear model is exact only within a modulation chart around a bias phase; Section~\ref{sec:multi_chart} makes the multi-chart continuation explicit, but its numerical evaluation on hardware-in-the-loop testbeds is deferred. These limitations do not invalidate the curvature coordinate; they specify where hardware and electromagnetic details enter.

\paragraph*{Directions deferred to future work.} Three items are framed here as future work rather than as gaps in the present paper:
\begin{itemize}\itemsep 2pt
    \item \emph{Learned regularization operator} \(\mathbf L\). The specification \(\mathbf L=\mathbf D_2\) in \eqref{eq:regest} is analytical. Replacing \(\mathbf L\) with a learned graph-Laplacian or attention-weighted regularizer on the aperture graph — trained on measured phase-noise data from a specific hardware platform — is a natural machine-learning contribution that preserves the closed-form bias--variance decomposition of Proposition~\ref{prop:bv}. This is a promising joint direction with holographic MIMO hardware groups.
    \item \emph{Wideband generalization} \cite{Verdu2002Spectral}. The propagation operator \(\cP\) is treated as narrowband. In wideband holographic MIMO the operator becomes frequency-dependent, \(\cP=\cP(f)\), and the tangent channel decomposes as \(\cT(f)=\cRop(f)\cP(f)M_{j\sqrt{P/L}\,e^{j\phi_{\mathrm b}(f)}}\cS\); the gauge-fixed synthesis operator \(\cS\) is frequency-independent by construction. A coherent-wideband capacity of the form \(C_{\mathrm WB}=\int_{0}^{B}\log_{2}\detF\!\bigl(\I+\widetilde\cT(f)\cK(f)\widetilde\cT(f)^{\ast}\bigr)\,df\) is a plausible conjecture, upper-bounded (per frequency) by the narrowband Fredholm capacity of Theorem~\ref{thm:fredholm}. Rigorous treatment is deferred to a companion paper.
    \item \emph{Vector electromagnetic extension}. Replacing the scalar propagation operator \(\cP\) with a vector electromagnetic operator (accounting for polarization, mutual coupling, and vector boundary conditions on the aperture surface) is a natural next step for connecting the coordinate theory here to full EIT. The gauge-fixing and Fredholm-capacity machinery of Sections~\ref{sec:model}--\ref{sec:capacity_infinite} transfers verbatim once the operator is compact on the vector \(L^{2}\) space.
\end{itemize}

\section{Numerical Diagnostics and Reproducibility}
\label{sec:simulations}

This section reports deterministic diagnostics for the analytical mechanisms
developed above. Every quantity is computed from the operators the theorems are
about: the gauge operator of Section~\ref{sec:model}, the tangent channel of
Theorem~\ref{thm:linearization} assembled on a paraxial Fresnel kernel, the
dual-budget water-filling of Corollary~\ref{cor:gwf}, and the receiver-side
second-difference operator of Theorem~\ref{thm:curvnoise}. No spectrum is
tabulated, no target value is solved for backwards, and no residual is floored
by an additive constant.

\subsection{Discretisation and its verification}

The gauge-fixed space is discretised in the orthonormal shifted-Legendre basis
\(\widehat P_k(x)=\sqrt{(2k+1)/L}\,P_k(2x/L-1)\). Two properties make this the
natural choice and both are checked rather than assumed. First, the \(L^{2}\)
mass matrix is the identity, so every eigenproblem below is a standard symmetric
problem rather than an ill-scaled pencil. Second --- and this is the load-bearing
point --- the affine gauge is spanned \emph{exactly} by \(\widehat P_0\)
(constant) and \(\widehat P_1\) (proportional to \(x-L/2\)), and by orthogonality
no other basis element meets it. Imposing the gauge is therefore the deletion of
two coefficients: there is no projector, no pseudo-inverse and no boundary
stencil, and consequently no boundary-induced bias in the computed constants.
The synthesis operator \(\cS\) is then assembled by integrating Legendre
coefficients twice in exact recursion, so the construction involves no quadrature
at all.

Two independent discretisations are carried, and their agreement is reported: the
exact synthesis route above, and a Gauss--Legendre quadrature of the curvature
form \(K_{kl}=\inner{\widehat P_k''}{\widehat P_l''}\). They share no code path.
For the tangent channel, the Fresnel kernel is integrated on Gauss--Legendre
nodes and the node count is refined until the resolvable modal gains stop moving;
the worst relative change across node counts is \(\SimQuadratureResidual\).
Modal gains below \(10^{-12}\) of the leading gain are not resolved by a
double-precision singular-value decomposition and are excluded from every
reported quantity; \(\SimResolvableModes\) modes survive that criterion for the
reference geometry.

\begin{table}[t]
\centering
\caption{Deterministic diagnostic settings.}
\label{tab:diagnostic_params}
\renewcommand{\arraystretch}{1.12}
\begin{tabular}{ll}
\toprule
Quantity & Setting \\
\midrule
Random seed & \(\SimSeed\) \\
Legendre degree (gauge operator) & \(400\) \\
Spectrum modes reported & \(128\) \\
Aperture lengths \(L_{\mathrm t}=L_{\mathrm r}\) & \(0.2\)~m \\
Link distance \(d\) & \(1.25\)~m \\
Wavelength \(\lambda\) & \(1\)~mm \\
Fresnel number & \(8.0\) \\
Shannon number \(N_{\mathrm S}=L_{\mathrm t}L_{\mathrm r}/(\lambda d)\) & \(\SimShannonNumber\) \\
Gauss--Legendre nodes per aperture & \(1200\) \\
Curvature-power budget \(P_{\mathrm c}\) & \(\SimCapBudget\) \\
\bottomrule
\end{tabular}
\end{table}

\subsection{Modal Spectrum Against Its Closed Form}
\label{sec:diag_spectrum}

Proposition~\ref{prop:Qspec_closed_form} gives the modal Gram spectrum in closed
form, \(\rho_m=(L/\beta_m)^{4}\). The diagnostic therefore does not fit a power
law --- it compares the computed spectrum to that expression. Over \(128\) modes
the largest relative deviation is \(\SimSpectrumMaxRelErr\), and the computed
Poincar\'{e}--Wirtinger constant \(\SimWirtingerConstant\) agrees with
\(L^{2}/\beta_1^{2}\) to a relative error of \(\SimWirtingerRelErr\), with
\(\beta_1=\SimBetaOne\).

Fig.~\ref{fig:diag_spectrum} shows the spectrum together with two \emph{null
runs}, whose purpose is to establish that the agreement above is informative
rather than automatic.

\begin{itemize}\itemsep 2pt
\item \emph{Null run 1: the Dirichlet sine prediction.} The argument rejected in
Remark~\ref{rem:sine_trap} predicts \(\rho_m=(m\pi/L)^{-4}\). The median ratio of
the true spectrum to that prediction is \(\SimNullSineRatio\), and the two differ
by a full factor of \(\beta_1^{4}/\pi^{4}\) at the leading mode --- so a
diagnostic that agreed with one would visibly disagree with the other. The
corresponding ratio on the Wirtinger constant, \(\beta_1^{2}/\pi^{2}\), is
computed by the same run as \(\SimWirtingerSlackFactor\); this is the factor by
which the value rejected in Remark~\ref{rem:sine_trap} overstates \(C_L\).
\item \emph{Null run 2: a plausible surrogate.} Substituting
\(\rho_m=(1+0.015\sin 0.7m)(m-\tfrac12)^{-4}\) and fitting a power law returns a
slope of \(\SimNullSurrogateSlope\), which is closer to \(-4\) than the true
spectrum's own fitted slope of \(\SimSlopeRho\) over the same window. The
surrogate nonetheless overstates the true spectrum by a factor of
\(\SimNullSurrogateRatio\) at \(m=128\) and inverts the sign of the half-integer
offset. This is the concrete reason a fitted exponent cannot serve as
verification here.
\end{itemize}

Consistently with Remark~\ref{rem:Qspec_decay}, the fitted slope of the
\emph{computed} spectrum over \(m\in[16,128]\) is \(\SimSlopeRho\) and the fitted
slope of the closed form over the identical window is \(\SimSlopeRhoExact\). They
agree; neither is \(-4\), and neither should be.

\begin{figure}[t]
\centering
\includegraphics[width=\columnwidth]{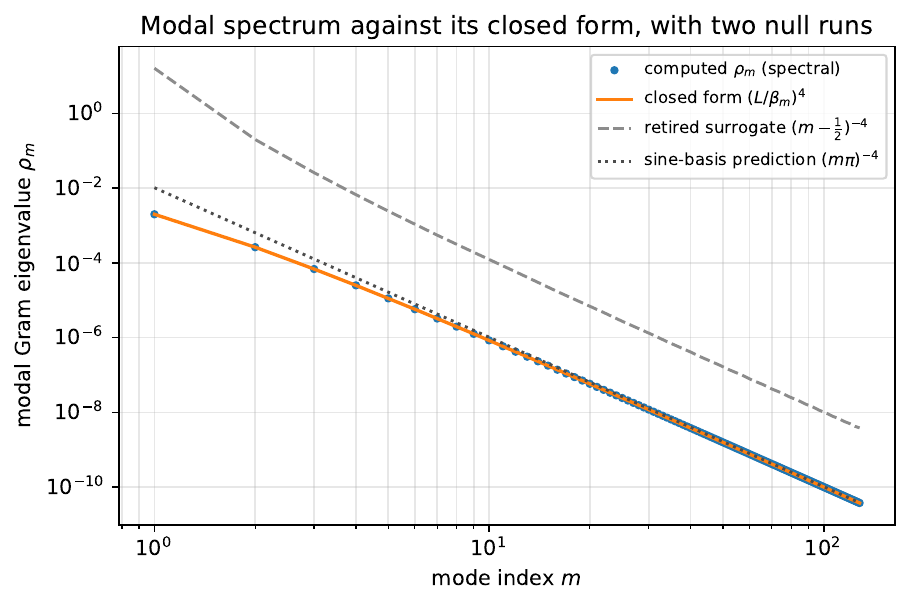}
\caption{Computed modal spectrum against the closed form
\(\rho_m=(L/\beta_m)^{4}\), with the two null runs of
Section~\ref{sec:diag_spectrum} overlaid. The rejected surrogate and the
Dirichlet sine prediction are visibly distinct from the operator's own
spectrum.}
\label{fig:diag_spectrum}
\end{figure}

\subsection{Tangent-Channel Remainder}

Theorem~\ref{thm:linearization} bounds the tangent remainder by
\(\norm{y_\epsilon(c)-\cT c}\le K_R\,\epsilon\,\norm{c}^{2}_{\cX}\), i.e.\ first
order in the modulation depth. The diagnostic measures exactly that quantity on
the assembled operator chain \(\cP M_{j\sqrt{P/L}e^{j\phi_{\mathrm b}}}\cS\), with
the curvature symbol normalised to unit peak phase excursion so that \(\epsilon\)
is in radians. The fitted slope is \(\SimSlopeLinearization\) and the empirical
constant is \(K_R=\SimRemainderConstant\).

Two points on method. The exponential is evaluated with \(\mathrm{expm1}\), and
the sweep starts at \(\epsilon=10^{-5}\): the remainder is \(O(\epsilon)\) while
the two terms differenced to obtain it are \(O(1)\), so roughly
\(\log_{10}(1/\epsilon)\) digits are lost and below that floor the measurement
would report cancellation rather than physics. Second, this is a different
quantity from the unnormalised Taylor residual of a scalar exponential, whose
\(O(\epsilon^{2})\) behaviour holds for any function and tests nothing about the
aperture. Fig.~\ref{fig:diag_linearization} plots the normalised remainder against
modulation depth on the Fresnel operator chain, together with the fitted slope.

\begin{figure}[t]
\centering
\includegraphics[width=\columnwidth]{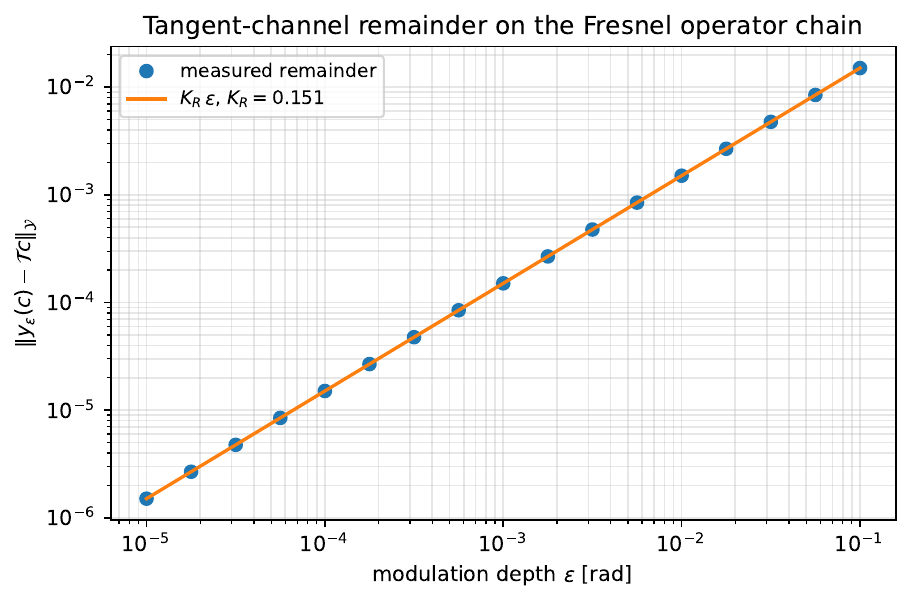}
\caption{Tangent-channel remainder on the Fresnel operator chain. The fitted
slope \(\SimSlopeLinearization\) confirms the first-order bound of
Theorem~\ref{thm:linearization}.}
\label{fig:diag_linearization}
\end{figure}

\subsection{Galerkin Truncation}

The curvature-power budget is fixed in advance at
\(P_{\mathrm c}=\SimCapBudget\) --- \(10^{4}\) times the inverse leading modal
gain --- and the capacity is whatever the truncation returns. At full truncation
it is \(\SimCapFinal\) bits per use, with a relative change of
\(\SimCapLastRelChange\) over the final truncation step. The reported number is a
measurement of the operator at a stated budget; it is not a target, and no budget
was solved backwards to reach it. The truncation sweep behind these numbers is
shown in Fig.~\ref{fig:diag_capacity}.

\begin{figure}[t]
\centering
\includegraphics[width=\columnwidth]{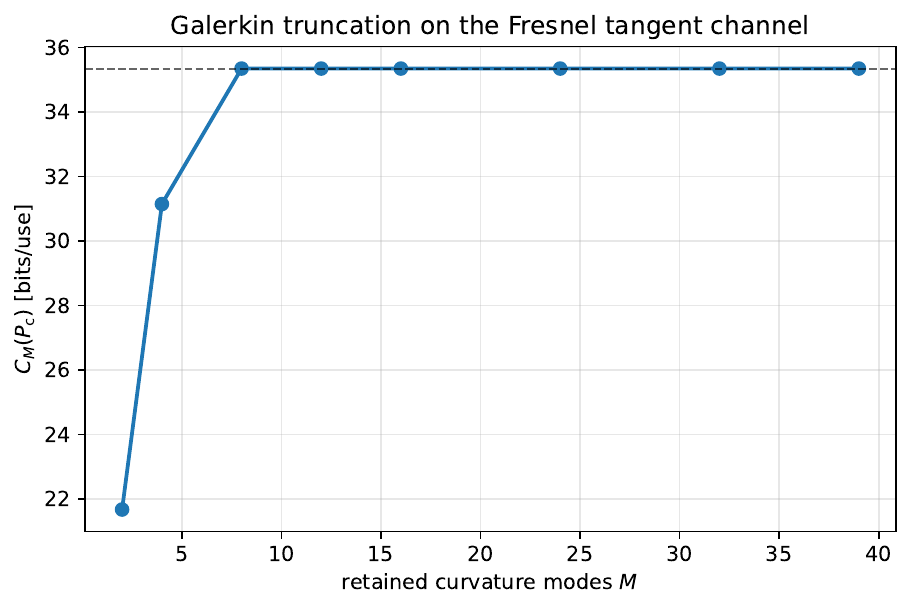}
\caption{Galerkin truncation of the finite-mode capacity on the Fresnel tangent
channel at the stated budget \(P_{\mathrm c}=\SimCapBudget\).}
\label{fig:diag_capacity}
\end{figure}

\subsection{Dual-Budget Generalized Water-Filling}
\label{sec:diag_dual}

Corollary~\ref{cor:gwf} and Algorithm~\ref{alg:gwf} are exercised here for the
first time in the regime that carries their content: both budgets simultaneously
active, with both multipliers strictly positive. The excursion budget is swept
from slack down to \(12\%\) of the excursion the unconstrained solution consumes.
At the tightest point the multipliers are \(\mu=\SimMuAtTightest\) and
\(\nu=\SimNuAtTightest\) --- comparable in magnitude, so the second constraint is
genuinely binding rather than numerically vestigial --- the active set has moved
by \(\SimDualBudgetActiveShift\) modes, from \(\SimActiveModesSlack\) to
\(\SimActiveModesTight\), and the rate given up is \(\SimDualBudgetRateLoss\).
The dual is solved by nested monotone bisection rather than by dual ascent, so
convergence does not depend on a step size. Across the whole sweep the maximum
measured KKT violation is \(\SimMaxKktResidual\) and the maximum measured duality
gap is \(\SimMaxDualGap\). Both are computed from the returned point; neither is
floored.

\emph{Null run.} Tightening any budget moves the active set, so a shift alone
would not establish that the second multiplier does anything new. The
discriminating comparison is against the classical single-budget solution at
whatever power level reproduces the \emph{same capacity}, which is
\(\SimNullEqualCapacityPower\). That allocation differs from the dual-budget one
by a relative \(L^{2}\) distance of \(\SimDualBudgetAllocationDivergence\), and
it would spend \(\SimNullExcursionRatio\) times as much phase excursion to
deliver the same rate. The second multiplier therefore reprices modes; it is not
a reparametrisation of the first. Fig.~\ref{fig:diag_dual_budget} shows the rate cost of
tightening the excursion budget alongside both multipliers across the same sweep.

\begin{figure}[t]
\centering
\includegraphics[width=\columnwidth]{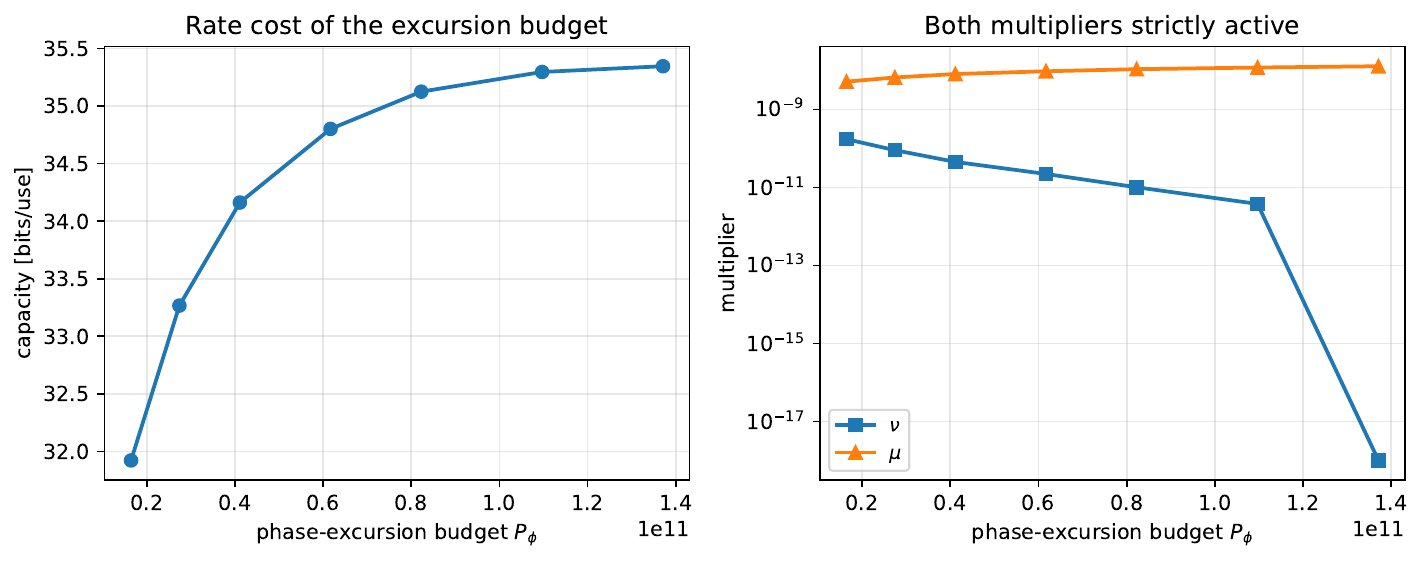}
\caption{Dual-budget generalized water-filling. Left: rate cost of tightening the
phase-excursion budget. Right: both multipliers across the same sweep, with
\(\nu>0\) throughout the constrained region.}
\label{fig:diag_dual_budget}
\end{figure}

\subsection{Derivative-Noise Amplification}

Theorem~\ref{thm:curvnoise} bounds the curvature-noise covariance by
\(\norm{\mathbf K_{\mathrm c}}_2\le 16\sigma_\phi^{2}\Delta x^{-4}\). Measuring
\(\norm{\mathbf D_2\mathbf D_2^{\TT}}_2\) across grids from \(n=16\) to
\(n=512\) gives a fitted exponent of \(\SimSlopeNoise\) in \(n\), and the ratio
of the measured norm to the analytic bound rises monotonically to
\(\SimNoiseBoundRatio\) --- approaching unity from below without ever exceeding
it. The bound is therefore asymptotically tight as well as valid, which is a
stronger statement than the exponent alone. Fig.~\ref{fig:diag_noise} plots the
measured norm against the analytic bound over the sampling sweep.

\begin{figure}[t]
\centering
\includegraphics[width=\columnwidth]{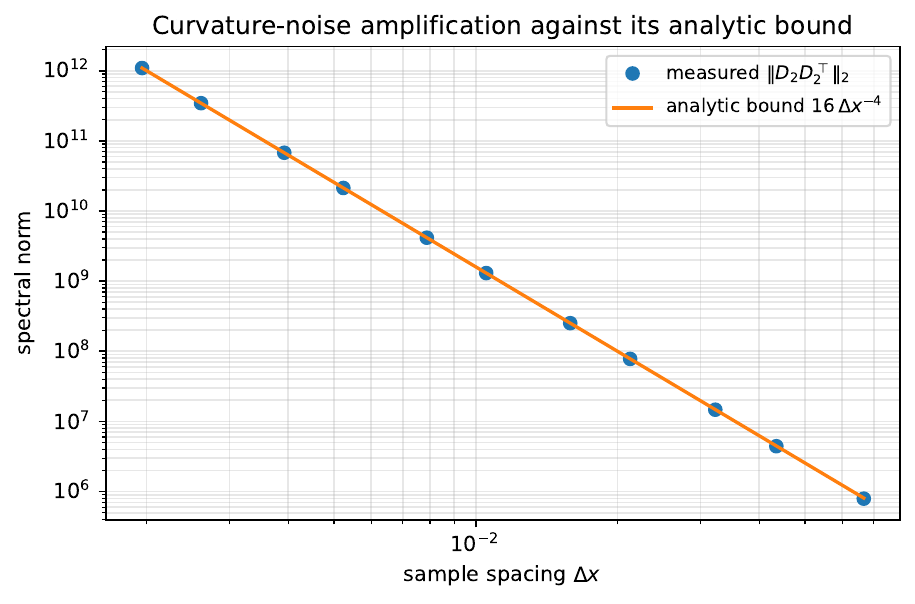}
\caption{Curvature-noise amplification against the analytic bound
\(16\,\Delta x^{-4}\). The measured norm approaches the bound from below.}
\label{fig:diag_noise}
\end{figure}

\subsection{The Two Scaling Exponents, and the Crossover}
\label{sec:diag_crossover}

Section~\ref{sec:crossover} distinguishes the per-mode and total-power forms of
the differentiation-limited law, and predicts where the polynomial branch ends.
Both are measured here on the same gain law, so the difference cannot be read as
a modelling artefact: the exponent in the per-mode SNR is
\(\SimSlopeRootPerMode\), and the exponent under a total-power budget with
water-filling is \(\SimSlopeRootTotalPower\).

On the Fresnel tangent channel at \(N_{\mathrm S}=\SimShannonNumber\), the
propagation operator has effective rank \(\SimPropagationRank\) at a \(10^{-6}\)
relative threshold. Sweeping the total power over eighteen decades, the local
exponent \(d\log C/d\log P\) averages \(\SimCrossoverPlateauExponent\) while the
active-mode count lies between \(8\) and \(N_{\mathrm S}\), and falls to at most
\(\SimCrossoverSaturatedExponent\) once that count saturates. This is the
crossover of Theorem~\ref{thm:crossover}, measured.

\emph{Null run.} Mode exhaustion would collapse the exponent for any truncated
gain sequence, so the collapse must be attributed to the propagation operator
rather than to the solver or to truncation. Repeating the identical sweep with
the propagation replaced by the identity and the mode supply made effectively
unlimited, the exponent stays at \(\SimNullFlatPropExponent\) across the same
power range in which the physical kernel has already fallen to
\(\SimCrossoverSaturatedExponent\). Both branches and the null run are shown in
Fig.~\ref{fig:diag_crossover}, and the two-sided envelope of Theorem~\ref{thm:composition}
that governs which of them binds is plotted in Fig.~\ref{fig:diag_composition}.

\begin{figure}[t]
\centering
\includegraphics[width=\columnwidth]{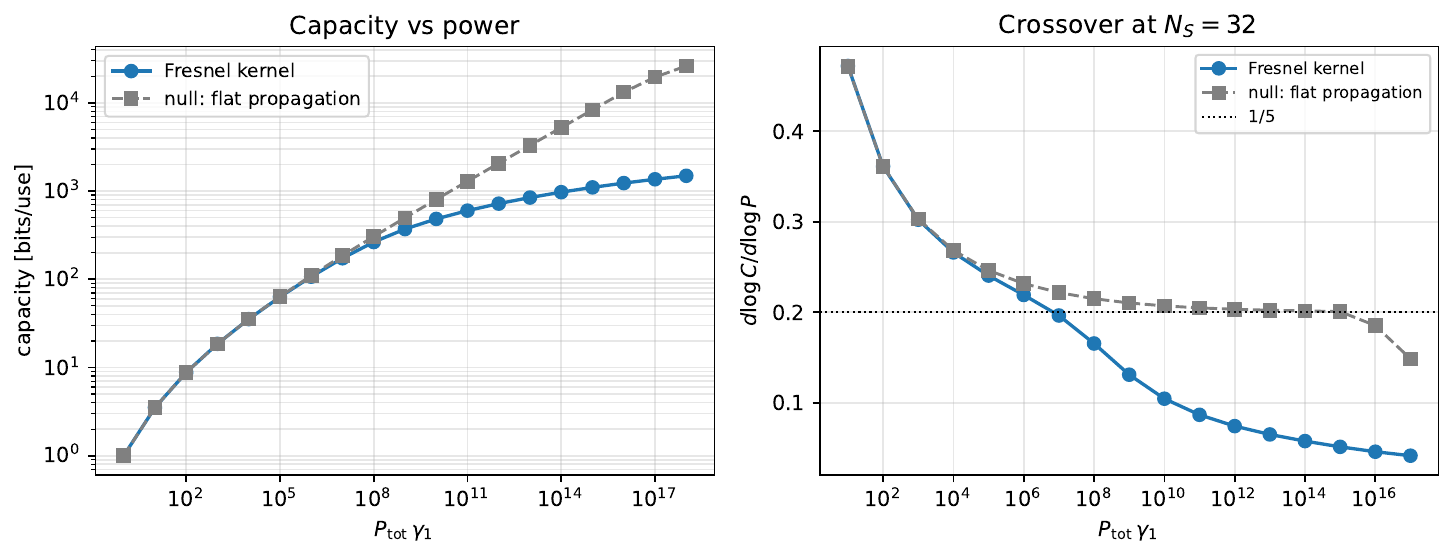}
\caption{Capacity and local exponent against total power on the Fresnel kernel,
with the flat-propagation null run overlaid. The polynomial branch ends where the
active-mode count reaches \(N_{\mathrm S}\).}
\label{fig:diag_crossover}
\end{figure}

\begin{figure}[t]
\centering
\includegraphics[width=\columnwidth]{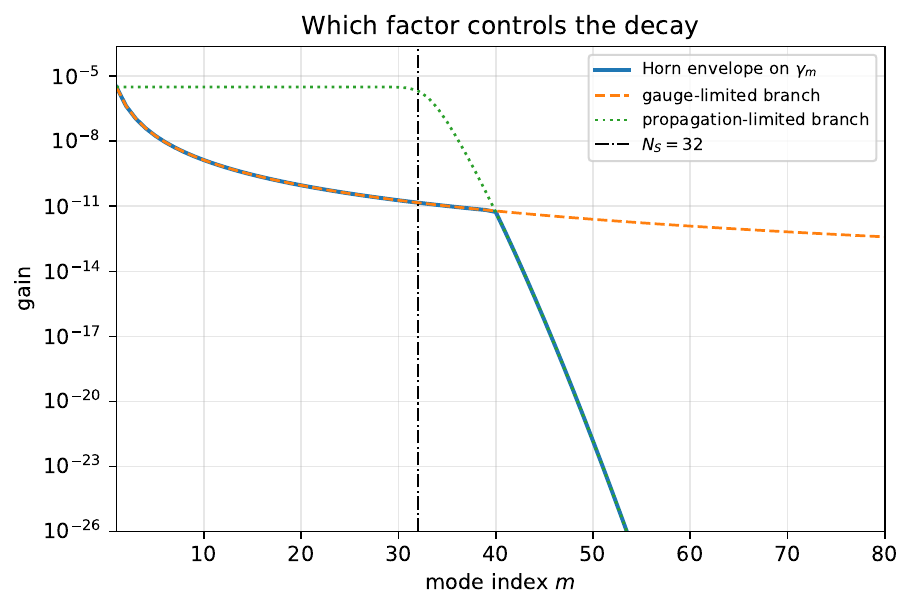}
\caption{The Horn--Weyl envelope of Theorem~\ref{thm:composition} against its two
branches. The gauge factor controls the decay below \(N_{\mathrm S}\); the
propagation factor controls it above.}
\label{fig:diag_composition}
\end{figure}

\subsection{What These Diagnostics Do and Do Not Establish}

They establish that the computed spectrum matches its closed form to
\(\SimSpectrumMaxRelErr\); that the tangent remainder obeys the first-order bound
of Theorem~\ref{thm:linearization} on a physical kernel; that the dual-budget
law of Corollary~\ref{cor:gwf} is solvable to a measured KKT residual of
\(\SimMaxKktResidual\) with both multipliers active, and that its allocation is
not reproducible by any single-budget solution at equal capacity; that the
derivative-noise bound is valid and asymptotically tight; and that the
differentiation-limited branch exists but terminates at the Shannon number.

They do not establish anything about a calibrated device. The propagation model
is scalar and paraxial; polarisation, mutual coupling, phase-shifter resolution
and calibration error are outside this layer. The two-dimensional Hessian-gauge
numerical support of Section~\ref{sec:2d} awaits the port of the same projector-
free spectral construction to the rectangle, where the Saint-Venant range
condition must be imposed on the curvature space rather than absorbed into a
coefficient deletion; no two-dimensional fitted exponents are claimed here in the
interim.

\section{Conclusion}

This paper introduced a curvature-domain signal model for continuous-aperture wireless communication. By quotienting phase profiles by affine piston-and-tilt gauge freedom, the second spatial derivative of phase becomes a complete coordinate for the gauge-fixed phase class. The inverse map from curvature to phase is compact, and this compact synthesis operator induces a modal representation in which phase-realizability weights decay with the fourth-order law associated with inverse second differentiation.

Starting from the exact nonlinear phase-only aperture law, we derived the coherent local tangent channel and showed that its compactness follows from gauge-fixed synthesis. This leads to finite-dimensional and infinite-dimensional capacity formulations under curvature and phase-excursion budgets. The finite-dimensional problem admits a KKT characterization and, in commuting cases, a generalized water-filling law. The infinite-dimensional formulation is expressed through a Fredholm determinant, admits an optimal covariance, and is recovered as the monotone limit of modal truncations.

The receiver-side analysis identifies the main practical bottleneck. Curvature estimation requires two spatial differentiations of phase, so unregularized curvature noise is amplified quartically with sampling resolution. This derivative-noise mechanism limits useful modal growth and leads to a fourth-root high-SNR capacity scaling in the differentiation-limited regime. Regularized curvature reconstruction therefore becomes an essential part of any practical architecture, not merely a numerical convenience.

The deterministic diagnostics of Section~\ref{sec:simulations} measure these mechanisms rather than illustrate them. The computed modal spectrum agrees with the closed form \(\rho_m=(L/\beta_m)^{4}\) to a maximum relative deviation of \(\SimSpectrumMaxRelErr\), and the fitted log--log slope over \(m\in[16,128]\) is \(\SimSlopeRho\) against the closed form's own \(\SimSlopeRhoExact\) on the identical window --- neither is \(-4\), and Remark~\ref{rem:Qspec_decay} explains why neither should be. The tangent remainder on a paraxial Fresnel chain has fitted slope \(\SimSlopeLinearization\), confirming the first-order bound of Theorem~\ref{thm:linearization}. At the stated budget \(P_{\mathrm c}=\SimCapBudget\) the Galerkin truncation converges to \(\SimCapFinal\) bits/use with a last-step change of \(\SimCapLastRelChange\). The dual-budget law is exercised with both multipliers strictly positive, to a measured KKT residual of \(\SimMaxKktResidual\) and duality gap of \(\SimMaxDualGap\). The derivative-noise diagnostic gives fitted exponent \(\SimSlopeNoise\) in the sample count and approaches the analytic bound to \(\SimNoiseBoundRatio\) from below. Finally, the differentiation-limited exponent is \(\SimSlopeRootPerMode\) in per-mode SNR and \(\SimSlopeRootTotalPower\) under a total-power budget, and on the physical kernel it holds at \(\SimCrossoverPlateauExponent\) only until the active-mode count reaches \(N_{\mathrm S}=\SimShannonNumber\), after which it collapses to \(\SimCrossoverSaturatedExponent\).

The framework is deliberately conservative. It does not introduce new electromagnetic physics and does not claim capacity beyond the full aperture-field channel. Its role is to expose a gauge-invariant coordinate system inside ordinary phase-coded aperture propagation and to quantify both its benefits and its derivative-noise costs. Future work should add propagation-specific numerical studies, including calibrated Fresnel and full-wave models, two-dimensional hardware-constrained surfaces, phase wrapping, finite phase resolution, mutual coupling, polarization, and adaptive regularization strategies that jointly optimize curvature excitation and noise-robust reconstruction.

\appendices

\section{Proof of Proposition \ref{prop:quotient} and Theorem \ref{thm:synthesis}}

\subsection*{Proof of Proposition \ref{prop:quotient}}
If \(\phi_1-\phi_2\in \cA\), then \((\phi_1-\phi_2)''=0\), so \(\phi_1''=\phi_2''\); hence the map \([\phi]\mapsto \phi''\) is well defined. Injectivity follows because \(\phi''=0\) implies \(\phi\in \cA\). Surjectivity follows from the construction in the proof below of Theorem \ref{thm:synthesis}, which gives a representative in \(\cV\) for any prescribed curvature \(c\in L^2(\Omega)\). \hfill\(\square\)

\subsection*{Proof of Theorem \ref{thm:synthesis}}
\emph{Injectivity:}
If \(\phi\in\cV\) and \(\phi''=0\), then \(\phi(x)=a+b(x-L/2)\). Since \(\phi\in\cV\),
\[
\inner{\phi}{\psi_0}_{\cX}=aL=0,
\qquad
\inner{\phi}{\psi_1}_{\cX}=b\norm{\psi_1}_{\cX}^2=0,
\]
hence \(a=b=0\), so \(\phi=0\).

\emph{Surjectivity:}
Given \(c\in \cX\), define
\[
\phi_0(x)=\int_0^x (x-t)c(t)\,dt.
\]
Then \(\phi_0\in H^2(\Omega)\) and \(\phi_0''=c\). Set
\[
a=-\frac{1}{L}\inner{\phi_0}{\psi_0}_{\cX},
\qquad
b=-\frac{1}{\norm{\psi_1}_{\cX}^2}\inner{\phi_0}{\psi_1}_{\cX},
\]
and define
\[
\phi=\phi_0+a\psi_0+b\psi_1.
\]
Then \(\phi''=c\) and \(\phi\in\cV\). Thus \(D^2:\cV\to\cX\) is bijective.

\emph{Boundedness and compactness:}
The map \(D^2:\cV\to\cX\) is bounded because it is the second derivative restricted to \(H^2(\Omega)\). Since it is bijective between Banach spaces, \(\cS=(D^2)^{-1}\) is bounded by the open mapping theorem. Further, \(\cS:\cX\to H^2(\Omega)\) is bounded, while the embedding \(H^2(\Omega)\hookrightarrow L^2(\Omega)\) is compact on bounded intervals by Rellich's theorem \cite{Evans}. Hence \(\cS:\cX\to\cX\) is compact. Equation \eqref{eq:synthesis_bound} follows from boundedness. \hfill\(\square\)

\section{Proof of Theorem \ref{thm:linearization} and Corollary \ref{cor:compact}}

Let \(u=\cS c\). By Theorem \ref{thm:synthesis}, \(\cS:\cX\to H^2(\Omega)\) is bounded. Since \(H^2(\Omega)\hookrightarrow L^\infty(\Omega)\) continuously in one dimension, there exists \(C_{\infty}>0\) such that
\[
\norm{u}_{L^\infty}\le C_{\infty}\norm{c}_{\cX}.
\]
Now
\[
s_\epsilon(c)-s_\epsilon(0)
=
\sqrt{\frac{P}{L}}\,e^{j\phi_{\mathrm b}}
\bracks{
e^{j\epsilon u}-1
}.
\]
Hence
\[
\frac{s_\epsilon(c)-s_\epsilon(0)}{\epsilon}
=
j\sqrt{\frac{P}{L}}\,e^{j\phi_{\mathrm b}}u
+
\sqrt{\frac{P}{L}}\,e^{j\phi_{\mathrm b}}
\frac{e^{j\epsilon u}-1-j\epsilon u}{\epsilon}.
\]
For real \(t\), \(\abs{e^{jt}-1-jt}\le t^2/2\). Therefore,
\[
\abs{
\frac{e^{j\epsilon u(x)}-1-j\epsilon u(x)}{\epsilon}
}
\le
\frac{\epsilon}{2}\abs{u(x)}^2.
\]
Taking \(L^2\)-norms and using \(\norm{u^2}_{L^2}\le \norm{u}_{L^\infty}\norm{u}_{L^2}\),
\begin{align*}
&\norm{
\frac{s_\epsilon(c)-s_\epsilon(0)}{\epsilon}
-
j\sqrt{\frac{P}{L}}\,e^{j\phi_{\mathrm b}}u
}_{\cU_t}\\
&\qquad\le
C\,\epsilon\,\norm{u}_{L^\infty}\norm{u}_{L^2}
\le
K_R\,\epsilon\,\norm{c}_{\cX}^2
\end{align*}
for \(\norm{c}_{\cX}\le R\). Applying bounded \(\cP\) and \(\cRop\) yields
\[
\norm{e_\epsilon(c)}_{\cY}\le K_R\,\epsilon\,\norm{c}_{\cX}^2,
\]
which proves \eqref{eq:first_order_chan}--\eqref{eq:remainder} with
\[
\cT=\cRop \cP M_{j\sqrt{P/L}\,e^{j\phi_{\mathrm b}}}\cS.
\]
Since \(\cS:\cX\to\cX\) is compact and multiplication by \(e^{j\phi_{\mathrm b}}\), \(\cP\), and \(\cRop\) are bounded, their composition is compact. If \(\cP\) has square-integrable kernel \(G\), it is Hilbert--Schmidt and hence compact. \hfill\(\square\)

\section{Proof of Proposition \ref{prop:Qspec}, Proposition \ref{prop:matrixchan}, and Theorem \ref{thm:galerkin}}

\subsection*{Proof of Proposition \ref{prop:Qspec}}
Since \(\cS\) is compact, \(\cQ=\cS^\ast\cS\) is compact, self-adjoint, and positive. The spectral theorem for compact self-adjoint operators gives an orthonormal eigenbasis \(\{u_m\}\) for the closure of the range together with eigenvalues \(\rho_m>0\) decreasing to zero. Since \(D^2:\cV\to\cX\) is onto, \(\cS\) is injective and \(\cQ\) has trivial kernel; hence the eigenvectors form an orthonormal basis of \(\cX\). Expanding \(c=\sum_m a_m u_m\) gives \(\norm{c}_{\cX}^2=\sum_m |a_m|^2\). Moreover,
\[
\norm{\cS c}_{\cX}^2
=
\inner{\cS c}{\cS c}_{\cX}
=
\inner{c}{\cQ c}_{\cX}
=
\sum_{m\ge 1}\rho_m |a_m|^2.
\]
This proves the claim. \hfill\(\square\)

\subsection*{Proof of Proposition \ref{prop:matrixchan}}
Substitute \(c_M=\sum_{j=1}^M a_j u_j\) into \(y=\cT c+z\), project against \(v_i\), and use linearity:
\[
[\mathbf y_M]_i
=
\sum_{j=1}^M
\inner{v_i}{\cT u_j}_{\cY}a_j
+
\inner{v_i}{R_M z}_{\cY}.
\]
This yields \eqref{eq:matrixchan}--\eqref{eq:ndef}. The covariance formula follows by definition. Finally, \(\norm{\cS c_M}_{\cX}^2=\sum_{m=1}^M\rho_m |a_m|^2=\mathbf a_M^\HH \mathbf G_M\mathbf a_M\), giving \eqref{eq:GM}. \hfill\(\square\)

\subsection*{Proof of Theorem \ref{thm:galerkin}}
For any \(c\in \cX\),
\[
\norm{\cT c - R_M \cT P_M c}_{\cY}
\le
\norm{\cT - R_M \cT P_M}_{\cL(\cX,\cY)} \norm{c}_{\cX}.
\]
Squaring and taking expectations yields \eqref{eq:truncbound}. It remains to prove \eqref{eq:normconv}. Decompose
\[
\cT-R_M\cT P_M
=
(I-R_M)\cT
+
R_M\cT(I-P_M).
\]
Because \(\cT\) is compact, \(\cT(B_{\cX})\) is relatively compact in \(\cY\), and strong convergence \(R_M\to I\) upgrades to uniform convergence on this compact set:
\[
\norm{(I-R_M)\cT}_{\cL(\cX,\cY)}\to 0.
\]
Similarly, compactness of \(\cT\) and strong convergence \(P_M\to I\) imply
\[
\norm{\cT(I-P_M)}_{\cL(\cX,\cY)}\to 0.
\]
Since \(\norm{R_M}\le 1\), \eqref{eq:normconv} follows. \hfill\(\square\)

\section{Proof of Theorem \ref{thm:capacity}, Proposition \ref{prop:KKT}, Corollary \ref{cor:gwf}, and Proposition \ref{prop:capub}}

Whiten the channel:
\[
\widetilde{\mathbf y}_M
=
\mathbf K_{n,M}^{-1/2}\mathbf y_M
=
\widetilde{\mathbf H}_M \mathbf a_M + \widetilde{\mathbf n}_M,
\qquad
\widetilde{\mathbf H}_M=\mathbf K_{n,M}^{-1/2}\mathbf H_M,
\]
with \(\widetilde{\mathbf n}_M\sim \CN(\mathbf 0,\I)\). For fixed covariance \(\mathbf K_{a,M}\), Gaussian input maximizes differential entropy \cite{CoverThomas,Telatar}, which yields
\[
I(\mathbf a_M;\widetilde{\mathbf y}_M)
=
\log_2\det\!\paren{
\I + \widetilde{\mathbf H}_M \mathbf K_{a,M}\widetilde{\mathbf H}_M^\HH
}.
\]
The feasible set is convex, closed, bounded, and finite-dimensional, hence compact; the objective is continuous and concave. This proves Theorem \ref{thm:capacity}.

For Proposition \ref{prop:KKT}, write the Lagrangian
\begin{align*}
&\mathcal L(\mathbf K,\mu,\nu,\mathbf Z)
=
\log_2\det\!\Bigl(
\I+\mathbf K_{n,M}^{-1/2}\mathbf H_M \mathbf K\\
&\qquad\times\mathbf H_M^\HH \mathbf K_{n,M}^{-1/2}
\Bigr)
-\mu(\Tr \mathbf K-P_{\mathrm c})\\
&\qquad-\nu(\Tr(\mathbf G_M \mathbf K)-P_{\phi})
+\Tr(\mathbf Z \mathbf K),
\end{align*}
with \(\mathbf Z\succeq 0\). Differentiating gives
\begin{align*}
\frac{\partial \mathcal L}{\partial \mathbf K}
&=
\frac{1}{\ln 2}
\mathbf H_M^\HH
\paren{
\mathbf K_{n,M}+\mathbf H_M\mathbf K\mathbf H_M^\HH
}^{-1}
\mathbf H_M\\
&\quad-\mu\I-\nu\mathbf G_M+\mathbf Z.
\end{align*}
Setting this to zero at the optimum yields \eqref{eq:KKTineq}--\eqref{eq:KKTcomp}.

For Corollary \ref{cor:gwf}, joint diagonalize \(\mathbf B_M\) and \(\mathbf G_M\). The problem becomes
\[
\max_{p_m\ge 0}\;
\sum_{m=1}^{M}\log_2(1+\gamma_m p_m)
\]
subject to
\[
\sum_{m=1}^{M}p_m\le P_{\mathrm c},
\qquad
\sum_{m=1}^{M}\rho_m p_m\le P_\phi.
\]
Scalar KKT conditions yield
\[
\frac{\gamma_m}{(1+\gamma_m p_m)\ln 2}\le \mu+\nu\rho_m,
\]
with equality on active modes, which gives \eqref{eq:gwf}.

For Proposition \ref{prop:capub}, feasibility implies
\[
\rho_M\Tr(\mathbf K_{a,M})
\le
\Tr(\mathbf G_M\mathbf K_{a,M})
\le
P_\phi,
\]
and therefore
\[
\Tr(\mathbf K_{a,M})\le P_{\mathrm{eff},M}.
\]
Whitening gives
\[
C_M^\star
=
\max_{\mathbf K\succeq 0}
\log_2\det\!\paren{\I+\widetilde{\mathbf H}_M \mathbf K \widetilde{\mathbf H}_M^\HH}
\]
with \(\Tr(\mathbf K)\le P_{\mathrm{eff},M}\). Since \(\rank(\widetilde{\mathbf H}_M)\le r_M\), Jensen's inequality gives
\begin{align*}
&\log_2\det(\I+\widetilde{\mathbf H}_M \mathbf K \widetilde{\mathbf H}_M^\HH)\\
&\qquad\le
r_M \log_2\!\paren{
1+\frac{1}{r_M}\Tr(\widetilde{\mathbf H}_M \mathbf K \widetilde{\mathbf H}_M^\HH)
}.
\end{align*}
Now
\begin{align*}
\Tr(\widetilde{\mathbf H}_M \mathbf K \widetilde{\mathbf H}_M^\HH)
&=
\Tr(\mathbf B_M \mathbf K)
\le
\lambda_{\max}(\mathbf B_M)\Tr(\mathbf K)\\
&\le
P_{\mathrm{eff},M}\Tr(\mathbf B_M).
\end{align*}
Combining the inequalities proves \eqref{eq:capub}. \hfill\(\square\)

\section{Proof of Theorems \ref{thm:fredholm}, \ref{thm:attainment_kkt}, \ref{thm:infwf}, and \ref{thm:CMtoCinfty}}

\subsection*{Proof of Theorem \ref{thm:fredholm}}
If \(\cK\in\mathfrak K\), then \(\cK\) is positive trace-class. Since \(\tilde{\cT}\) is bounded,
\[
\tilde{\cT}\cK\tilde{\cT}^\ast
\]
is positive trace-class, so \(\detF(I+\tilde{\cT}\cK\tilde{\cT}^\ast)\) is well defined \cite{Simon}. Also,
\[
\Tr(\tilde{\cT}\cK\tilde{\cT}^\ast)
=
\Tr(\tilde{\cT}^\ast\tilde{\cT}\cK)
\le
\norm{\tilde{\cT}}^2\Tr(\cK)
\le
\norm{\tilde{\cT}}^2 P_{\mathrm c}.
\]
Since \(\log_2\detF(I+A)\le \Tr(A)/\ln 2\) for \(A\succeq 0\), the capacity is finite.

Let
\[
\cK_N=\sum_{m=1}^{N}\lambda_m \psi_m\otimes \psi_m
\]
be the spectral truncation of \(\cK\). Then \(\cK_N\to \cK\) in trace norm, hence
\[
\tilde{\cT}\cK_N\tilde{\cT}^\ast \to \tilde{\cT}\cK\tilde{\cT}^\ast
\qquad
\text{in trace norm.}
\]
For finite \(N\), the input is \(N\)-dimensional Gaussian and its mutual information is
\[
\log_2\det\!\paren{I+\tilde{\cT}\cK_N\tilde{\cT}^\ast}.
\]
Passing to the trace-norm limit yields
\[
I(c;\tilde y)=\log_2\detF\!\paren{I+\tilde{\cT}\cK\tilde{\cT}^\ast}.
\]
Taking the supremum over \(\mathfrak K\) proves \eqref{eq:fredholm_capacity}. \hfill\(\square\)

\subsection*{Proof of Theorem \ref{thm:attainment_kkt}}
The positive trace-class ball \(\{\cK\succeq 0:\Tr(\cK)\le P_{\mathrm c}\}\) is compact in the weak-\(\ast\) topology induced by identifying trace-class operators with the dual of compact operators \cite{Conway,Simon}. Let \(\{\cK_n\}\subset\mathfrak K\) be a maximizing sequence. Passing to a subnet if necessary, \(\cK_n\) converges weak-\(\ast\) to some positive trace-class \(\cK^\star\).

For any finite-rank orthogonal projector \(E_J\),
\[
\Tr(E_J\cK^\star)=\lim_n \Tr(E_J\cK_n)\le P_{\mathrm c}.
\]
Taking the supremum over \(J\) gives \(\Tr(\cK^\star)\le P_{\mathrm c}\). Since \(\cQ\) is compact,
\[
\Tr(\cQ\cK^\star)=\lim_n \Tr(\cQ\cK_n)\le P_\phi.
\]
Thus \(\cK^\star\in\mathfrak K\).

It remains to show continuity of the objective along this weak-\(\ast\) convergence. Since \(\tilde{\cT}\) is compact, choose finite-rank \(\tilde{\cT}_J\) with
\[
\norm{\tilde{\cT}-\tilde{\cT}_J}\to 0.
\]
For fixed \(J\), the map
\[
\cK\mapsto
\log_2\det\!\paren{
I+\tilde{\cT}_J\cK\tilde{\cT}_J^\ast
}
\]
depends only on finitely many weak-\(\ast\)-continuous matrix entries and is therefore weak-\(\ast\) continuous. Moreover, uniformly over \(\Tr(\cK)\le P_{\mathrm c}\),
\begin{align*}
&
\norm{
\tilde{\cT}\cK\tilde{\cT}^{\ast}
-
\tilde{\cT}_J\cK\tilde{\cT}_J^{\ast}
}_1
\\
&\quad\le
\norm{\tilde{\cT}-\tilde{\cT}_J}\,P_{\mathrm c}\,\norm{\tilde{\cT}}
+
\norm{\tilde{\cT}_J}\,P_{\mathrm c}\,\norm{\tilde{\cT}-\tilde{\cT}_J}.
\end{align*}
The Fredholm log-determinant is Lipschitz on positive trace-class perturbations in the sense that
\[
\abs{\log_2\detF(I+A)-\log_2\detF(I+B)}
\le
\frac{\norm{A-B}_1}{\ln 2}
\]
for \(A,B\succeq 0\). Hence the finite-rank objectives converge uniformly to the full objective. The full objective is therefore weak-\(\ast\) continuous on the feasible set, and \(\cK^\star\) attains the supremum.

For the KKT system, the objective is concave and Fr\'echet differentiable on the positive trace-class cone, with derivative at \(\cK^\star\)
\[
D F(\cK^\star)
=
\frac{1}{\ln 2}\,
\tilde{\cT}^{\ast}
\paren{
I+\tilde{\cT}\cK^\star \tilde{\cT}^{\ast}
}^{-1}
\tilde{\cT}.
\]
The feasible set has affine trace constraints and the positive cone constraint. Standard convex KKT conditions yield multipliers \(\mu,\nu\ge 0\) and a positive normal-cone operator such that \eqref{eq:opKKTineq}--\eqref{eq:opKKTcomp} hold. Complementary slackness gives \eqref{eq:opslack1}--\eqref{eq:opslack2}. \hfill\(\square\)

\subsection*{Proof of Theorem \ref{thm:infwf}}
Because \(\cB\) and \(\cQ\) are commuting compact self-adjoint operators, they admit a common orthonormal eigenbasis \(\{u_m\}\). Let \(\widehat{\cK}\) be the matrix of \(\cK\) in this basis, with diagonal entries \(k_{mm}\). Then
\[
\Tr(\cK)=\sum_{m\ge 1}k_{mm},
\qquad
\Tr(\cQ\cK)=\sum_{m\ge 1}\rho_m k_{mm}.
\]
Moreover,
\[
I+\tilde{\cT}\cK\tilde{\cT}^\ast
\sim
I+\cB^{1/2}\cK\cB^{1/2},
\]
where the nonzero eigenvalues are the same. In the common eigenbasis, finite-dimensional Hadamard inequalities applied to principal truncations and then passed to the monotone trace-class limit give
\[
\detF(I+\cB^{1/2}\cK\cB^{1/2})
\le
\prod_{m\ge 1}(1+\beta_m k_{mm}),
\]
with equality when \(\widehat{\cK}\) is diagonal. Since the constraints depend only on the diagonal entries, an optimal covariance can be chosen diagonal:
\[
\cK=\sum_{m\ge 1} p_m\, u_m\otimes u_m.
\]
The optimization then reduces to \eqref{eq:Cinf_diag}--\eqref{eq:p_phase}. Scalar KKT conditions yield \eqref{eq:inf_gwf}. \hfill\(\square\)

\subsection*{Proof of Theorem \ref{thm:CMtoCinfty}}
Since \(\mathfrak K_M\subset \mathfrak K_{M+1}\subset \mathfrak K\), the sequence \(C_M\) is monotone increasing and bounded above by \(C_\infty\). Hence \(\lim_{M\to\infty} C_M\) exists and is at most \(C_\infty\).

Fix \(\varepsilon>0\). Choose \(\cK_\varepsilon\in\mathfrak K\) such that
\[
\log_2\detF\!\paren{I+\tilde{\cT}\cK_\varepsilon\tilde{\cT}^\ast}
\ge
C_\infty-\varepsilon.
\]
Define
\[
\cK_{\varepsilon,M}=P_M \cK_\varepsilon P_M.
\]
Because \(P_M\to I\) strongly and \(\cK_\varepsilon\) is trace-class,
\[
\cK_{\varepsilon,M}\to \cK_\varepsilon
\qquad
\text{in trace norm.}
\]
Also,
\[
\Tr(\cK_{\varepsilon,M})\le \Tr(\cK_\varepsilon)\le P_{\mathrm c}.
\]
Since \(P_M\) is the spectral projector of \(\cQ\), it commutes with \(\cQ\), and
\[
\Tr(\cQ \cK_{\varepsilon,M})
=
\Tr(P_M \cQ^{1/2}\cK_\varepsilon \cQ^{1/2} P_M)
\le
\Tr(\cQ\cK_\varepsilon)
\le P_\phi.
\]
Hence \(\cK_{\varepsilon,M}\in\mathfrak K_M\). By bounded sandwiching,
\[
\tilde{\cT}\cK_{\varepsilon,M}\tilde{\cT}^\ast
\to
\tilde{\cT}\cK_\varepsilon\tilde{\cT}^\ast
\qquad
\text{in trace norm,}
\]
and therefore
\[
\log_2\detF\!\paren{I+\tilde{\cT}\cK_{\varepsilon,M}\tilde{\cT}^\ast}
\to
\log_2\detF\!\paren{I+\tilde{\cT}\cK_\varepsilon\tilde{\cT}^\ast}.
\]
So for sufficiently large \(M\),
\[
C_M
\ge
\log_2\detF\!\paren{I+\tilde{\cT}\cK_{\varepsilon,M}\tilde{\cT}^\ast}
\ge
C_\infty-2\varepsilon.
\]
Since \(C_M\le C_\infty\), we conclude \(C_M\to C_\infty\). \hfill\(\square\)

\section{Proof of Theorem \ref{thm:curvnoise} and Proposition \ref{prop:bv}}

\subsection*{Proof of Theorem \ref{thm:curvnoise}}
From \eqref{eq:chat},
\[
\hat{\mathbf c}
=
\mathbf D_2 \bm{\phi} + \mathbf D_2 \mathbf w_\phi.
\]
Hence \(\E[\hat{\mathbf c}]=\mathbf D_2\bm{\phi}\), so the estimator is unbiased for the discrete second difference. Its covariance is
\[
\mathbf K_{\mathrm c}
=
\E[(\mathbf D_2\mathbf w_\phi)(\mathbf D_2\mathbf w_\phi)^\TT]
=
\sigma_\phi^2 \mathbf D_2\mathbf D_2^\TT.
\]
Each row of \(\mathbf D_2\) is \(\Delta x^{-2}[1,-2,1]\) shifted by one position. Therefore,
\[
[\mathbf K_{\mathrm c}]_{ii}
=
\frac{\sigma_\phi^2}{\Delta x^4}(1^2+(-2)^2+1^2)
=
\frac{6\sigma_\phi^2}{\Delta x^4},
\]
\[
[\mathbf K_{\mathrm c}]_{i,i\pm 1}
=
\frac{\sigma_\phi^2}{\Delta x^4}\big((-2)\cdot 1 + 1\cdot (-2)\big)
=
-\frac{4\sigma_\phi^2}{\Delta x^4},
\]
and
\[
[\mathbf K_{\mathrm c}]_{i,i\pm 2}
=
\frac{\sigma_\phi^2}{\Delta x^4}.
\]
All other overlaps vanish. Finally,
\[
\norm{\mathbf D_2}_2
\le
\frac{\norm{[1,-2,1]}_1}{\Delta x^2}
=
\frac{4}{\Delta x^2},
\]
so
\[
\norm{\mathbf K_{\mathrm c}}_2
=
\sigma_\phi^2 \norm{\mathbf D_2\mathbf D_2^\TT}_2
=
\sigma_\phi^2 \norm{\mathbf D_2}_2^2
\le
\frac{16\sigma_\phi^2}{\Delta x^4}.
\]
\hfill\(\square\)

\subsection*{Proof of Proposition \ref{prop:bv}}
The normal equations of \eqref{eq:regest} yield
\[
\paren{
\mathbf S_M^\TT \mathbf S_M
+
\lambda \sigma_\phi^2 \mathbf L^\TT \mathbf L
}
\hat{\mathbf a}_\lambda
=
\mathbf S_M^\TT \widetilde{\bm{\phi}},
\]
hence \eqref{eq:regsol}. If \(\widetilde{\bm{\phi}}=\mathbf S_M \mathbf a + \mathbf w\), then
\[
\hat{\mathbf a}_\lambda
=
\mathbf R_\lambda \mathbf S_M^\TT \mathbf S_M \mathbf a
+
\mathbf R_\lambda \mathbf S_M^\TT \mathbf w.
\]
Since
\[
\mathbf R_\lambda \mathbf S_M^\TT \mathbf S_M
=
\I-\lambda \sigma_\phi^2 \mathbf R_\lambda \mathbf L^\TT \mathbf L,
\]
the bias is \eqref{eq:bias}. The covariance is
\[
\Cov(\hat{\mathbf a}_\lambda)
=
\mathbf R_\lambda \mathbf S_M^\TT \Cov(\mathbf w)\mathbf S_M \mathbf R_\lambda
=
\sigma_\phi^2
\mathbf R_\lambda
\mathbf S_M^\TT \mathbf S_M
\mathbf R_\lambda,
\]
which proves \eqref{eq:covreg}. Summing squared bias and variance gives \eqref{eq:mse}. \hfill\(\square\)

\section{Proof of Proposition \ref{prop:pep}}

Whiten the channel:
\[
\widetilde{\mathbf y}
=
\mathbf K_{n,M}^{-1/2}\mathbf y_M
=
\widetilde{\mathbf H}_M \mathbf a + \widetilde{\mathbf n},
\qquad
\widetilde{\mathbf n}\sim \CN(\mathbf 0,\I).
\]
Let
\[
\bm{\delta}
=
\widetilde{\mathbf H}_M(\mathbf a-\mathbf b).
\]
An ML error \(\mathbf a\to \mathbf b\) occurs when
\[
\norm{\widetilde{\mathbf y}-\widetilde{\mathbf H}_M\mathbf b}_2^2
\le
\norm{\widetilde{\mathbf y}-\widetilde{\mathbf H}_M\mathbf a}_2^2,
\]
which is equivalent to
\[
\RePart\{\bm{\delta}^\HH \widetilde{\mathbf n}\}
\ge
\frac{\norm{\bm{\delta}}_2^2}{2}.
\]
Since \(\widetilde{\mathbf n}\sim \CN(\mathbf 0,\I)\), the scalar
\[
\RePart\{\bm{\delta}^\HH \widetilde{\mathbf n}\}
\sim \mathcal N\!\paren{0,\frac{\norm{\bm{\delta}}_2^2}{2}}.
\]
Hence
\[
P(\mathbf a\to \mathbf b\mid \mathbf H_M)
=
Q\!\paren{\frac{\norm{\bm{\delta}}_2}{\sqrt{2}}}.
\]
Finally,
\[
\norm{\bm{\delta}}_2^2
=
(\mathbf a-\mathbf b)^\HH
\mathbf H_M^\HH \mathbf K_{n,M}^{-1}\mathbf H_M
(\mathbf a-\mathbf b)
=
d_M^2(\mathbf a,\mathbf b),
\]
which proves \eqref{eq:pep} and the exponential bound \eqref{eq:pepbound}. \hfill\(\square\)

\section{Proof of Proposition \ref{prop:useful}, Theorem \ref{thm:rootlaw}, and Corollary \ref{cor:rootlaw_sharp}}

\subsection*{Proof of Proposition \ref{prop:useful}}
The condition \(\mathrm{snr}_m\ge \Gamma\) is equivalent to
\[
\frac{\alpha p}{\sigma_0^2+\beta \lambda_m^2}\ge \Gamma
\iff
\beta \lambda_m^2 \le \frac{\alpha p}{\Gamma}-\sigma_0^2.
\]
Since \(\lambda_m=(m\pi/L)^2\),
\[
\beta \paren{\frac{m\pi}{L}}^4
\le
\frac{\alpha p}{\Gamma}-\sigma_0^2.
\]
Solving for \(m\) gives \eqref{eq:muse}, with the positive-floor convention accounting for the case \(\alpha p/\Gamma-\sigma_0^2\le 0\). \hfill\(\square\)

\subsection*{Proof of Theorem \ref{thm:rootlaw}}
Define
\[
f_\Gamma(x)=\log_2\!\paren{1+\frac{\Gamma}{1+\beta x^4}}, \qquad x\ge 0.
\]
This function is positive and decreasing.

\emph{Lower bound:}
For sufficiently large \(\Gamma\), let
\[
m_0 = \left\lfloor \paren{\frac{\Gamma}{2\beta}}^{1/4}\right\rfloor.
\]
For \(1\le m\le m_0\),
\[
1+\beta m^4 \le 1+\frac{\Gamma}{2},
\]
so
\[
\frac{\Gamma}{1+\beta m^4}
\ge
\frac{\Gamma}{1+\Gamma/2}
\ge \frac{2}{3},
\qquad \Gamma\ge 1.
\]
Therefore,
\[
C(\Gamma)
\ge
m_0 \log_2\!\paren{1+\frac{2}{3}}
\ge
c_1'(\beta)\Gamma^{1/4}
\]
for all sufficiently large \(\Gamma\). On the compact interval of remaining \(\Gamma\)-values in \([1,\Gamma_0]\), the continuous positive function \(C(\Gamma)/\Gamma^{1/4}\) has a positive minimum. Combining the two constants gives a global \(c_1(\beta)>0\).

\emph{Upper bound:}
Since \(f_\Gamma\) is decreasing,
\[
C(\Gamma)
\le
f_\Gamma(0) + \int_0^\infty f_\Gamma(x)\,dx.
\]
Substitute \(x=\Gamma^{1/4}t\):
\begin{align*}
\int_0^\infty f_\Gamma(x)\,dx
&=
\Gamma^{1/4}
\int_0^\infty
\log_2\!\paren{
1+\frac{\Gamma}{1+\beta \Gamma t^4}
}\,dt \\
&\le
\Gamma^{1/4}
\int_0^\infty
\log_2\!\paren{
1+\frac{1}{\beta t^4}
}\,dt.
\end{align*}
The integral is finite because the integrand is logarithmic near zero and behaves as \(t^{-4}\) at infinity. Also \(f_\Gamma(0)=\log_2(1+\Gamma)\le C\Gamma^{1/4}\) for \(\Gamma\ge 1\). Hence
\[
C(\Gamma)\le c_2(\beta)\Gamma^{1/4}
\]
for some \(c_2(\beta)<\infty\). This proves \eqref{eq:rootlaw} and \eqref{eq:theta}. \hfill\(\square\)

\subsection*{Proof of Corollary \ref{cor:rootlaw_sharp}}
Retain \(f_\Gamma(x)=\log_2(1+\Gamma/(1+\beta x^{4}))\) from the previous
proof; it is positive and strictly decreasing on \([0,\infty)\). For such a
function the integral test gives
\[
\int_0^{\infty} f_\Gamma(x)\,dx-f_\Gamma(0)
\;\le\;
\sum_{m\ge 1} f_\Gamma(m)
\;\le\;
\int_0^{\infty} f_\Gamma(x)\,dx ,
\]
so that \(\abs{C(\Gamma)-\int_0^\infty f_\Gamma}\le f_\Gamma(0)=\log_2(1+\Gamma)\).
Substituting \(x=(\Gamma/\beta)^{1/4}u\),
\[
\int_0^{\infty} f_\Gamma(x)\,dx
=
\paren{\frac{\Gamma}{\beta}}^{1/4}
\int_0^{\infty}\log_2\!\paren{1+\frac{\Gamma}{1+\Gamma u^{4}}}\,du .
\]
For every \(u>0\) the quantity \(\Gamma/(1+\Gamma u^{4})\) increases in
\(\Gamma\) and is bounded above by \(u^{-4}\), with limit \(u^{-4}\). The
integrand therefore increases pointwise to \(\log_2(1+u^{-4})\) and is dominated
by it, so by monotone convergence
\[
\int_0^{\infty}\log_2\!\paren{1+\frac{\Gamma}{1+\Gamma u^{4}}}du
\;\longrightarrow\;
\frac{1}{\ln 2}\int_0^{\infty}\ln\!\paren{1+u^{-4}}du .
\]
The remaining integral is classical and elementary. Integrating by parts with
\(u=\ln(1+x^{-n})\) and \(dv=dx\), the boundary term \(x\ln(1+x^{-n})\)
vanishes at both ends for \(n>1\), leaving
\[
\int_0^{\infty}\ln\!\paren{1+x^{-n}}dx
= n\int_0^{\infty}\frac{dx}{1+x^{n}}
= \frac{\pi}{\sin(\pi/n)} ,
\]
using the standard value \(\int_0^\infty (1+x^{n})^{-1}dx=(\pi/n)/\sin(\pi/n)\).
At \(n=4\) this is \(\pi/\sin(\pi/4)=\pi\sqrt2\). Since
\(\log_2(1+\Gamma)=o(\Gamma^{1/4})\), dividing by \((\Gamma/\beta)^{1/4}\)
yields \eqref{eq:rootlaw_sharp}. \hfill\(\square\)

\section{Proof of Theorem \ref{thm:2d_synthesis}}

If \(\nabla^2\phi=0\) in distributions on the connected domain \(\Omega_2\), then \(\phi\) is affine:
\[
\phi(x)=a+b_1x_1+b_2x_2.
\]
If also \(\phi\in\cV_2\), orthogonality to \(\psi_0,\psi_1,\psi_2\) forces \(a=b_1=b_2=0\), after re-centering by \(\bar x_1,\bar x_2\). Thus \(\nabla^2:\cV_2\to\cC_2\) is injective. It is surjective by the definition \(\cC_2=\Range(\nabla^2|_{\cV_2})\).

Boundedness of \(\nabla^2:\cV_2\to L^2(\Omega_2;\mathrm{Sym}(2))\) follows directly from the definition of the \(H^2\)-norm. Conversely, the second-order Poincar\'e--Bramble--Hilbert inequality on bounded Lipschitz domains gives a constant \(C_{\Omega_2}\) such that
\[
\norm{\phi}_{H^2(\Omega_2)}
\le
C_{\Omega_2}\norm{\nabla^2\phi}_{L^2(\Omega_2)}
\qquad
\forall \phi\in\cV_2.
\]
Therefore the inverse \(\cS_2=(\nabla^2)^{-1}:\cC_2\to\cV_2\) is bounded. Finally, \(\cS_2:\cC_2\to H^2(\Omega_2)\) is bounded and the embedding \(H^2(\Omega_2)\hookrightarrow L^2(\Omega_2)\) is compact by Rellich's theorem. Hence \(\cS_2:\cC_2\to L^2(\Omega_2)\) is compact. The affine-invariance statement follows because affine functions have zero Hessian, and equality of Hessians implies the difference has zero Hessian and is therefore affine. \hfill\(\square\)

\balance

\end{document}